\documentclass[envcountsame,envcountsect]{llncs}

\usepackage[utf8x]{inputenc}
\usepackage{amssymb}
\usepackage{amsmath}
\usepackage{paralist}
\usepackage[noend]{algorithm2e}
\usepackage{marginnote}
\usepackage{subfigure}
\usepackage{nicefrac}
\usepackage{ifthen}
\usepackage{xspace}
\usepackage{tikz}
\usepackage{wrapfig}

\newcommand{\nintrel}{\rightarrowtail}

\newcommand{\dintrel}[1]{\nintrel_{#1}}
\newcommand{\Actions}{A}
\newcommand{\Dom}{D}
\newcommand{\States}{S} 

\newcommand{\agts}{\ensuremath{D}}
\newcommand{\dom}{{\tt dom}} 
\newcommand{\obs}{{\tt obs}}
\newcommand{\wrt}{w.\,r.\,t.\xspace}
\newcommand{\redlogm}{\ensuremath{\leq_{m}^{\log}}}
\newcommand{\colorsys}[1]{\ensuremath{C(#1)}}
\newcommand{\diffcols}[2]{\ensuremath{E(#1,#2)}}
\newcommand{\colorsysprime}[1]{\ensuremath{C'(#1)}}
\newcommand{\diffcolsprime}[2]{\ensuremath{E'(#1,#2)}}
\newcommand{\last}{\ensuremath{last}\xspace}
\newcommand{\lastprime}{\ensuremath{last'}\xspace}
\newcommand{\boxSpacing}{}

\newcommand{\sources}{{\tt src}}
\newcommand{\ipurge}{{\tt ipurge}}
\newcommand{\purge}{{\tt purge}}

\newcommand{\dpurge}{{\tt purge}}

\newcommand{\dpsecure}{t-secure\xspace}
\newcommand{\dpsecty}{t-security\xspace}
\newcommand{\dpsectysection}{t-Security\xspace}

\newcommand{\dipsecure}{i-secure\xspace}
\newcommand{\dipsecurity}{i-security\xspace}
\newcommand{\dipsecuritysection}{i-Security\xspace}
\newcommand{\ipsecurity}{IP-security\xspace}
\newcommand{\ipsecuritysection}{IP-Security\xspace}

\newcommand{\dpsecurity}{\dpsecty}

\newcommand{\unwind}{\sim}

\newcommand{\dpu}{\textnormal{t}}
\newcommand{\ocdp}{\textnormal{(OC$_\dpu$)}\xspace}
\newcommand{\scdp}{\textnormal{(SC$_\dpu$)}\xspace}
\newcommand{\lrdp}{\textnormal{(LR$_\dpu$)}\xspace}

\newcommand{\dip}{\textnormal{i}}
\newcommand{\gOC}{\textnormal{(OC$_\dip$)}\xspace}
\newcommand{\gSC}{\textnormal{(SC$_\dip$)}\xspace}
\newcommand{\gLR}{\textnormal{(LR$_\dip$)}\xspace}

\newcommand{\uniform}{\textnormal{u}}
\newcommand{\dOC}{\textnormal{(OC$_{\dip}^{\uniform}$)}\xspace}
\newcommand{\dPC}{\textnormal{(PC$_\dip^\uniform$)}\xspace}
\newcommand{\dSC}{\textnormal{(SC$_\dip^\uniform$)}\xspace}
\newcommand{\dLR}{\textnormal{(LR$_\dip^\uniform$)}\xspace}

\newcommand{\dsrc}[3]{\ensuremath{{\tt src}(#1,#2,#3)}}
\newcommand{\dipurge}[3]{\ensuremath{{\tt ipurge}(#1,#2,#3)}}
\newcommand{\dipurgename}{\ensuremath{{\tt ipurge}}}
\newcommand{\dipurgeprime}[3]{\ensuremath{{\tt ipurge'}(#1,#2,#3)}}
\newcommand{\dipurgeprimename}{\ensuremath{{\tt ipurge'}}}
\newcommand{\actions}{\ensuremath{A}}
\newcommand{\set}[1]{\ensuremath\left\{#1\right\}}
\newcommand{\mathtext}[1]{\ensuremath{\mathrm{\text{#1}}}}
\newcommand{\indistinguishable}{\approx}
\newcommand{\card}[1]{\left| #1 \right|}
\newcommand{\complexityclassname}[1]{\ensuremath{\mathrm{#1}}}
\newcommand{\NP}{\complexityclassname{NP}}
\newcommand{\PTIME}{\complexityclassname{P}}
\newcommand{\NL}{\complexityclassname{NL}}
\newcommand{\infagents}[2]{\ensuremath{#1^{\leftarrowtail}_{#2}}}
\newcommand{\direl}[1]{\precsim_{#1}}
\newcommand{\dpol}{\ensuremath{(\dintrel s)_{s\in S}}\xspace}
\newcommand{\dpolprime}{\ensuremath{(\dintrel{s}')_{s\in S}}\xspace}
\newcommand{\beforeOpSpace}{\ }
\newcommand{\dpsimilar}{t-similar\xspace}
\newcommand{\dpsimilarity}{t-similarity\xspace}
\newcommand{\dipsimilar}{i-similar\xspace}
\newcommand{\dipsimilarity}{i-similarity\xspace}

\renewenvironment{itemize}{\begin{compactitem}}{\end{compactitem}}
\renewenvironment{enumerate}{\begin{compactenum}}{\end{compactenum}}

\usetikzlibrary{arrows,automata,fit,patterns}
\usetikzlibrary{shapes.geometric,shapes.arrows,decorations.pathmorphing}
\usetikzlibrary{matrix,chains,scopes,positioning,arrows,fit}
\usetikzlibrary{shapes,shadows,calc}
\usepgflibrary{arrows}

\tikzset{tikzglobal/.style={
    ->,
    >=stealth',
    shorten >=1pt,
    auto,
    node distance=0.8cm, 
    every path/.style=semithick, 
   initial text={}
 }}
 
 \tikzset{transnodedistance/.style={
  node distance=15mm
  }
}

\tikzset{transsysstate/.style={
	  draw,
    rounded corners,
    thick,
    outer sep=2pt
  }}
  
  \tikzset{systemstate/.style={%
    matrix of nodes,
    draw=black, 
    line width=1pt,
    rectangle,
   inner sep=2pt,
    rounded corners
  }}
  
  \tikzset{localpolicy/.style={%
    matrix of nodes,
    draw=black, 
    line width=1pt,
    rectangle,
   inner sep=2pt,
   dashed,
    node distance=10mm
  }}

\tikzset{agent/.style={
  draw=none
}}

\tikzset{policy/.style={ 
   >->,
   >=to
   }
}

\tikzset{policyedge/.style={
   >->,
   >=to
   }
}
\tikzset{graybox/.style={
		fill=black!15,
		rounded corners
   }
}

\pgfdeclarelayer{grayboxes}
\pgfsetlayers{grayboxes,main}

\title{Noninterference with Local Policies} 
\author{Sebastian Eggert \and Henning Schnoor \and Thomas Wilke}
\institute{Institut f\"ur Informatik, Christian-Albrechts-Universit\"{a}t zu Kiel, 24098 Kiel, Germany \email{\{sebastian.eggert|henning.schnoor|thomas.wilke\}@email.uni-kiel.de}}
\usetikzlibrary{positioning}

\begin{document}

\maketitle

\pagestyle{plain}

\begin{abstract}
We develop a theory for state-based noninterference in a setting where different security policies---we call them  local policies---apply in different parts of a given system. 
Our theory comprises appropriate security definitions, characterizations of these definitions, for instance in terms of unwindings, algorithms for analyzing the security of systems with local policies, and corresponding complexity results. 
\end{abstract}

\section{Introduction}

Research in formal security aims to provide rigorous definitions for different notions of security as well as methods to analyse a given system with regard to the security goals. Restricting the information that may be available to a user of the system (often called an agent) is an important topic in security. Noninterference~\cite{GogMes,goguen84} is a notion that formalizes this. Noninterference uses a security policy that specifies, for each pair of agents, whether information is allowed to flow from one agent to the other. To capture different aspects of information flow, a wide range of definitions of noninterference has been proposed, see, e.g.,~\cite{DBLP:conf/csfw/YoungB94,DBLP:conf/csfw/Millen90,DBLP:conf/esorics/Oheimb04,DBLP:conf/sp/WittboldJ90}. 

In this paper, we study systems where in different parts different policies apply. This is motivated by the fact that different security requirements may be desired in different situations, for instance, a user may want to forbid interference between his web browser and an instant messenger program while visiting banking sites but when reading a news page, the user may find interaction between these programs useful.

As an illustrating example, consider the system depicted in Fig.~\ref{fig:admin changes policy}, where three agents are involved: an administrator $A$ and two users $H$ and $L$. The rounded boxes represent system states, the arrows represent transitions. The labels of the states indicate what agent $L$ observes in the respective state; the labels of the arrows denote the action, either action $a$ performed by $A$ or action $h$ performed by $H$, inducing the respective transition. 
Every action can be performed in every state; if it does not change the state (i.\,e., if it induces a loop), the corresponding transition is omitted in the picture. 

The lower part of the system constitutes a secure subsystem with respect to the bottom policy: when agent $H$ performs the action $h$ in the initial state, the observation of agent $L$ changes from $0$ to $1$, but this is allowed according to the policy, as agent $H$ may interfere with agent $L$---there is an edge from $H$ to $L$. 

Similarly, the upper part of the system constitutes a secure subsystem with respect to the top policy: interference between $H$ and $L$ is not allowed---no edge from $H$ to $L$---and, in fact, there is no such interference, because $L$'s observation does not change when $h$ performs an action.

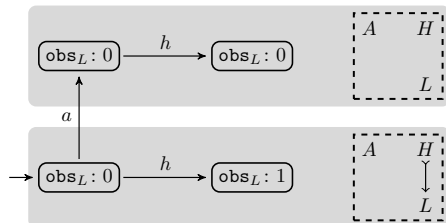
\begin{wrapfigure}[11]{r}{6.2cm}
\scalebox{0.75}{
\begin{tikzpicture}[tikzglobal]  

  \begin{scope}[transnodedistance]
    \node[transsysstate,initial] (q0) {$\obs_L\colon 0$} ;
    \node[transsysstate,above=of q0] (q2) {$\obs_L\colon 0$} ;
    \node[transsysstate,right=of q0] (q1) {$\obs_L\colon 1$} ;
    \node[transsysstate,above=of q1] (q3) {$\obs_L\colon 0$} ;
  \end{scope}
  \path (q0) edge node {$a$} (q2) ;
  \path (q0) edge node {$h$} (q1) ;
  \path (q2) edge node {$h$} (q3) ;

  \newcommand{\agents}{ 
    \node[agent] (A) [left] {$A$};
    \node[agent] (H) [right of=A]  {$H$};
    \node[agent] (L) [below of=H] {$L$};
  }
  \node[localpolicy] (q4) [right=of q1] {
    \agents
    \\
  };
  \path[policyedge] (H) edge (L);
  \node[localpolicy] (q5) [right=of q3] {
    \agents
    \\
  };
 
  \begin{pgfonlayer}{grayboxes}
  \node[graybox, fit= (q0) (q1) (q4)] (sub0) {};
  \node[graybox, fit= (q2) (q3) (q5)] (sub1) {};
  \end{pgfonlayer}
\end{tikzpicture}
}
\vspace*{-6mm}
\label{fig:admin changes policy}
\caption{System with local policies}
\end{wrapfigure}

However, the entire system is clearly insecure: agent $A$ must not interfere with anyone---there is no edge starting from $A$ in either policy---but when $L$ observes ``$1$'' in the lower right state, $L$ can conclude that $A$ did \emph{not} perform the $a$ action depicted.

Note that interference between $H$ and $L$ is allowed, unless $A$ performs action $a$.
But $L$ must not get to know whether $a$ was performed. To achieve this, interference between $H$ and $L$ must never be allowed. 
Otherwise, as we have just argued, 
$L$ can---by observing $H$'s actions---conclude that in the current part of the system, interference between $H$ and $L$ is still legal and thus 
$A$ did not perform $a$. 
In other words, in the policy of the lower part, 
the edge connecting $H$ and $L$ can never be ``used'' for an actual information flow. 
We call such edges \emph{useless}.---Useless edges are a key issue arising in systems with local policies. 

\paragraph{Our results.} We develop a theory of noninterference with local policies which takes the aforementioned issues into account.
Our contributions are as follows:
\begin{enumerate}
 \item We provide new and natural definitions for noninterference with local policies, both for the transitive~\cite{GogMes,goguen84} (agent $L$ may only be influenced by agent $H$ if there is an edge from $H$ to $L$ in the policy) and for the intransitive setting~\cite{HY87} (interference between $H$ and $L$ via ``intermediate steps'' is also allowed).
 \item We show that policies can always be rewritten into a normal form which does not contain any ``useless'' edges (see above).
 \item We provide characterizations of our definitions based on unwindings, which demonstrate the robustness of our definitions and from which we derive efficient verification algorithms.
 \item We provide results on the complexity of verifying noninterference. In the transitive setting, noninterference can be verified in nondeterministic logarithmic space (\NL). In the intransitive setting, the problem is \NP-complete, but fixed-parameter tractable with respect to the number of agents.
\end{enumerate}

Our results show significant differences between the transitive and the intransitive setting. In the transitive setting, one can, without loss of generality, always assume a policy is what we call uniform, which means that each agent may ``know'' (in a precise epistemic sense) the set of agents that currently may interfere with him. Assuming uniformity greatly simplifies the study of noninterference with local policies in the transitive setting. Moreover, transitive noninterference with local policies can be characterized by a simple unwinding, which yields very efficient algorithms. 

In the intransitive setting, the situation is more complicated. Policies cannot be assumed to be uniform, verification is \NP-complete, and, consequently, we only give an unwinding condition that requires computing exponentially many relations. However, for \emph{uniform} policies, the situation is very similar to the transitive setting: we obtain simple unwindings and efficient algorithms.

As a consequence of our results for uniform policies, we obtain an unwinding characterization of IP-security~\cite{HY87} (which uses a single policy for the entire system). Prior to our results, only an unwinding characterization that was \emph{sound}, but not \emph{complete} for IP-security was known~\cite{rushby92}. Our new unwinding characterization immediately implies that IP-security can be verified in nondeterministic logarithmic space, which improves the polynomial-time result obtained in~\cite{emsw11}.

\textit{Related Work}. Our intransitive security definitions generalize IP-security~\cite{HY87} mentioned above. The issues raised against IP-security in~\cite{meyden2007} are orthogonal to the issues arising from local policies. We therefore study local policies in the framework of IP-security, which is technically simpler than, e.g., TA-security as defined in~\cite{meyden2007}. 

Several extensions of intransitive noninterference have been discussed, for instance, in~\cite{RoscoeG99,DBLP:journals/jcs/MyersSZ06}. In~\cite{Leslie-DYNAMICNONINTERFERENCE-SSE-2006}, a definition of intransitive noninterference with local policies is given, however, the definition in~\cite{Leslie-DYNAMICNONINTERFERENCE-SSE-2006} does not take into account the aforementioned effects, and that work does not provide complete unwinding characterizations nor complexity results.

\section{State-based Systems with Local Policies}
\label{sect:preliminaries}

We work with the standard state-observed system model, that is, a system is a deterministic finite-state automaton where each action belongs to a dedicated agent and each agent has an observation in each state. More formally, a \emph{system} is a tuple $M=(S,s_0,\actions,\mathtt{step},\mathtt{obs},\mathtt{dom})$, where $S$ is a finite set of \emph{states}, $s_0\in S$ is the \emph{initial state}, $\Actions$ is a finite set of \emph{actions}, $\mathtt{step}\colon S\times\actions\rightarrow S$ is a \emph{transition function}, $\mathtt{obs}\colon S\times D\rightarrow O$ is an \emph{observation function}, where $O$ is an arbitrary set of observations, and $\mathtt{dom}\colon\actions\rightarrow D$ associates with each action an agent, where $D$ is an arbitrary finite set of agents (or security domains).

For a state $s$ and an agent $u$, we write $\obs_u(s)$ instead of $\obs(s,u)$. For a sequence $\alpha\in\actions^*$ of actions and a state $s\in S$, we denote by $s\cdot\alpha$ the state obtained when performing $\alpha$ starting in $s$, i.e., $s\cdot\epsilon=s$ and $s\cdot\alpha a=\mathtt{step}(s\cdot\alpha,a)$.

A \emph{local policy} is a reflexive relation
${\dintrel{}}\subseteq D\times D$. 
To keep our notation simple, we do not define subsystems nor policies for subsystems explicitly.
Instead, we assign a local policy to every state and denote the policy in state $s$ by $\dintrel{s}$.
We call the collection of all local policies $\dpol$ the \emph{policy} of the system.  
If $(u,v) \in {\dintrel{s}}$ for some $u, v \in \Dom$, $s \in \States$, we 
say $u \dintrel{s} v$ is an \emph{edge} in $\dpol$. 
A system has a \emph{global policy} if all local policies $\dintrel{s}$ are the
same in all states, i.e., if $u\dintrel sv$ does not depend on $s$. In this case, we denote the single policy by $\nintrel$ and only write $u \nintrel v$.
We define the set $\infagents us$ as  the set of agents that \emph{may interfere} with $u$ in $s$, i.e., the set $\set{v\ \vert\ v\dintrel su}$.

In the following, we fix an arbitrary system $M$ and a policy $\dpol$.

In our examples, we often identify a state with an action sequence leading to it from the initial state $s_0$, that is, we write $\alpha$ for $s_0 \cdot \alpha$, which is well-defined, because we consider deterministic systems. For example, in the system from Fig.~\ref{fig:admin changes policy}, we denote the initial state by $\epsilon$ and the upper right state by $ah$. In each state, we write the local policy in that state as a graph. In the system from Fig.~\ref{fig:admin changes policy}, we have $H\dintrel{\epsilon}L$, but $H\not\dintrel{a}L$. In general, we only specify the agents' observations as far as relevant for the example, which usually is only the observation of the agent $L$. We adapt the notation from Fig.~\ref{fig:admin changes policy} to our definition of local policies, which assigns a local policy to every state: we depict the graph of the local policy inside the rounded box for the state, see Fig.~\ref{fig:dpsecure_system}.

\section{The Transitive Setting}
\label{sect:transitive}

In this section, we define noninterference for systems with local policies in the transitive setting, give several characterizations, introduce the notion of useless edge, and discuss it. 
The basic idea of our security definition is that an occurrence of an action which, according to a local policy, should not be observable by an agent $u$ must not have any influence on $u$'s future observations.

\begin{definition}[\dpsecty]
  \label{def:local-dpurge-secure} 
  The system $M$ is \dpsecure iff for all $u \in \Dom$, $s \in \States$, $a \in \Actions$ and $\alpha \in \Actions^*$ the following implication holds:
  \begin{equation*}
    \text{If } \dom(a) \not \dintrel{s} u, \text{ then } 
    \obs_u(s \cdot \alpha) = \obs_u(s \cdot a \alpha) \enspace. 
  \end{equation*}
\end{definition}

\setlength\intextsep{0pt}
\begin{wrapfigure}[9]{r}{0pt}
\scalebox{0.75}{
  \begin{tikzpicture}[tikzglobal]
  
    \newcommand{\agents}{ 
      \node[agent] (A) [left] {$A$};
      \node[agent] (B) [right of=A] {$B$};
      \node[agent] (L) [below of=A] {$L$}; 
}

 \node[initial,systemstate] (q0)  {
   \agents
   \\
   \hline
   $\obs_L \colon 0$ \\
};
   \path (A) edge node {} (L)
   (B) edge node {} (L);

 \node[systemstate] (q1) [right=of q0] {
   \agents
  \\
   \hline
   $\obs_L \colon 1$ \\
};
   \path (B) edge node {} (L);

 \node[systemstate] (q2) [below=of q1] {
   \agents
  \\
   \hline
   $\obs_L \colon 2$ \\
};
   \path (A) edge node {} (L);

\path (q0) edge  node {$b$} (q1) 
edge  node {$a$} (q2)
(q1) edge [bend left=10] node {$b$} (q2)
(q2) edge [bend left=10] node  {$a$} (q1);
\end{tikzpicture}
}
\caption{A \dpsecure system}
\label{fig:dpsecure_system}
\end{wrapfigure}

Fig.~\ref{fig:dpsecure_system} shows a \dpsecure system. In contrast, the system in Fig.~\ref{fig:admin changes policy} is not \dpsecure, since $A \not \dintrel{\epsilon} L$, but $\obs_L(ah) \neq \obs_L(h)$.

\subsection{Characterizations of \dpsectysection}

In Theorem~\ref{thm:dp_characterizations}, we give two characterizations of \dpsecty, underlining that our definition is quite robust. The first characterization is based on an operator which removes all actions that must not be observed. It is essentially the definition from Goguen and Meseguer~\cite{GogMes,goguen84} of the $\purge$ operator generalized to systems with local policies.

\begin{definition}[purge for local policies]
  \label{sec:dynamic-transitive-purge} 
  For all $u \in \Dom$ and $s \in \States$ let $ \dpurge(\epsilon, u,
  s) = \epsilon$ and for all $a \in \Actions$ and $\alpha \in
  \Actions^*$ let 
  \begin{align*}
    \dpurge(a \alpha, u, s)  & = 
    \begin{cases}
      a \beforeOpSpace \dpurge(\alpha, u, s \cdot a) &\text{if } \dom(a) \dintrel{s} u
      \\
      \dpurge(\alpha, u, s) & \text{otherwise} \enspace. 
    \end{cases}
  \end{align*}
\end{definition}
The other characterization is in terms of unwindings, which we define for local policies in the following, generalizing the definition of Haigh and Young~\cite{HY87}.

\begin{definition}[transitive unwinding with local policies]
 A \emph{transitive unwinding} for $M$ with a policy $\dpol$ is a family of equivalence relations $(\unwind_u)_{u\in\Dom}$ such that for every agent $u \in \Dom$, all states $s, t \in \States$ and
all $a \in \Actions$, the following holds:
\begin{itemize}
\item If $\dom(a) \not\dintrel{s} u$, then $s \unwind_u s\cdot a$. \hfill \textnormal{ \lrdp---local respect}
\item If $s \unwind_u t$, then $s\cdot a \unwind_u t \cdot a$. \hfill \textnormal{\scdp---step consistency}
\item If $s \unwind_u t$, then $\obs_u(s) = \obs_u(t)$. \hfill \textnormal{\ocdp---output consistency}
\end{itemize}
\end{definition}

Our characterizations of \dpsecty are spelled out in the following theorem.

\begin{theorem}[characterizations of \dpsecty]
  \label{thm:dp_characterizations} 
  The following are equivalent:
  \begin{enumerate}
  \item The system $M$ is \dpsecure.
    \label{thm:dp_characterization_def}
  \item For all $u \in \Dom$, $s \in \States$, and $\alpha, \beta \in \Actions^*$ with $\dpurge(\alpha, u, s) = \dpurge(\beta, u, s)$, we have $\obs_u(s \cdot \alpha) = \obs_u(s \cdot \beta)$.
    \label{thm:dp_characterization_purge}
  \item There exists a transitive unwinding for $M$ with the policy $\dpol$.
    \label{thm:dp_characterization_unwind}
  \end{enumerate}
\end{theorem}

Unwinding relations yield efficient verification procedure. For verifying t-security, it is sufficient to compute for  every $u \in \Dom$ the smallest equivalence relation satisfying \lrdp and \scdp and check that the function $\obs_u$ is constant on every equivalence class. This can be done with nearly the same algorithm as is used for global policies, described in~\cite{emsw11}. The above theorem directly implies that \dpsecty can be verified in nondeterministic logarithmic space.

\subsection{Useless Edges}

An ``allowed'' interference $v\dintrel su$ may contradict a ``forbidden'' interference $v\not\dintrel{s'}u$ in a state $s'$ that should be indistinguishable to $s$ for $u$. 
In this case, the edge $v\dintrel su$ is useless. What this means is that an edge $v\dintrel su$ in the policy may be deceiving and should not be interpreted as ``it is allowed that $v$ interferes with $u$'', rather, it should be interpreted as ``it is not explicitly forbidden that $v$ interferes with $u$''. To formalize this, we introduce the following notion:

\begin{definition}[\dpsimilarity]
  States $s$, $s'$ are \emph{\dpsimilar} for an agent $u \in \Dom$, denoted $s \indistinguishable_u s'$, if there exist $t \in \States$, $a \in \Actions$, and $\alpha \in \Actions^*$ such that $\dom(a) \not\dintrel{t} u$, $s = t \cdot a \alpha$, and $s' = t\cdot \alpha$.
\end{definition}
Observe that \dpsimilarity is identical with the smallest equivalence relation satisfying \lrdp and \scdp. Also observe that the system $M$ is \dpsecure if and only if for every agent $u$, if $s\indistinguishable_u s'$, then $\obs_u(s)=\obs_u(s')$. 

The notion of \dpsimilarity allows us to formalize the notion of a useless edge:

\begin{definition}[useless edge]
  An edge $v\dintrel su$ is \emph{useless} if there is a state~$s'$ with $s\indistinguishable_u s'$ and  $v\not\dintrel{s'}u$.
\end{definition}

For example, consider again the system in Fig.~\ref{fig:admin changes policy}. Here, the local policy in the initial state allows information flow from $H$ to $L$. However, if $L$ is allowed to observe $H$'s action in the initial state, then $L$ would know that the system is in the initial state, and would also know that $A$ has not performed an action. This is an information flow from $A$ to $L$, which is prohibited by the policy. 

Useless edges can be removed without any harm:

\begin{theorem}[removal of useless edges]
  \label{theorem:uniformpolicies}
  Let $\dpolprime$ be defined by 
  \begin{align*}
    {\dintrel s'} = {\dintrel s} \setminus \{v \dintrel s u \mid v
    \dintrel s u \text{ is useless}\} \qquad \text{for all $s \in S$.}
  \end{align*}
  Then $M$ is \dpsecure \wrt $\dpol$ iff $M$ is \dpsecure \wrt $\dpolprime$.
\end{theorem}

The policy $\dpolprime$ in Theorem~\ref{theorem:uniformpolicies} has no useless edges, hence every edge in one of its local policies represents an \emph{allowed} information flow---no edge contradicts an edge in another local policy. Another interpretation is that any information flow that is \emph{forbidden} is \emph{directly} forbidden via the absence of the corresponding edge. In that sense, the policy is closed under logical deduction. 

We call a policy $\dpol$ \emph{uniform} if $\infagents u{s}=\infagents u{s'}$ holds for all states $s$ and $s'$ with $s\indistinguishable_u s'$. In other words, in states that $u$ should not be able to distinguish, the exact same set of agents may interfere with $u$. Hence $u$ may ``know'' the set of agents that currently may interfere with him. Note that a policy is uniform if and only if it does not contain useless edges. (This is not true in the intransitive setting, hence the seemingly complicated definition of uniformity.) Uniform policies have several interesting properties, for example, with a uniform policy the function $\dpurge$ behaves very similarly to the setting with a global policy: it suffices to verify action sequences that start in the initial state of the system and $\dpurge$ satisfies a natural associativity condition on a uniform policy. 

\section{The Intransitive Setting}
\label{sect:intransitive case}

In this section, we consider the intransitive setting, where,
whenever an agent performs an action, this event may transmit information about
the actions the agent has performed himself as well as information about
actions by other agents that was previously transmitted to him. The
definition follows a similar pattern as that of \dpsecty:
if performing an action sequence $a\alpha$ starting in a state $s$ should
not transmit the action $a$ (possibly via several intermediate steps)
to the agent $u$, then $u$ should be unable to deduce from his
observations whether $a$ was performed. 
To formalize this, 
we use Leslie's extension~\cite{Leslie-DYNAMICNONINTERFERENCE-SSE-2006} of Rushby's definition~\cite{rushby92} of ${\texttt{sources}}$. 

\begin{definition}[sources]
  For an agent $u$ let 
$\dsrc{\epsilon}us  =\set{u}$
and for 
$a \in \Actions$, $\alpha \in \Actions^*$, if
      $\dom(a) \dintrel{s} v$ for some $v \in \dsrc{\alpha}{u}{s\cdot
        a}$, then let
    $\dsrc{a\alpha}us = \dsrc{\alpha}{u}{s\cdot a} \cup \set{dom(a)}$, and else let $\dsrc{a\alpha}us = \dsrc{\alpha}{u}{s\cdot a}$.
\end{definition}

The set $\dsrc{a\alpha}us$ contains the agents that ``may know'' whether the action~$a$ has been performed in state $s$ after the run $a\alpha$ is performed: initially, this is only the set of agents $v$ with $\dom(a)\dintrel sv$. The knowledge may be spread by every action performed by an agent ``in the know:'' if an action $b$ is performed in a later state $t$, and $\dom(b)$ already may know that the action $a$ was performed, then all agents $v$ with $\dom(b)\dintrel{t}v$ may obtain this information when $b$ is performed.
Following the discussion above, we obtain a natural definition of security:
 
\begin{definition}[\dipsecurity] 
  The system $M$ is \dipsecure iff for all $s \in \States$, $a\in\Actions$, and $\alpha\in\Actions^*$, the following implication holds.
\begin{align*}
  \text{If $\dom(a)\notin\dsrc{a\alpha}{u}{s}$, then $\obs_u(s\cdot
    a\alpha)=\obs_u(s\cdot\alpha)$.}
\end{align*}
\end{definition}

The definition formalizes the above: if, on the path $a\alpha$, the action $a$ is not transmitted to $u$,
then $u$'s observation must not depend on whether $a$ was performed; the runs $a\alpha$ and $\alpha$ must be indistinguishable for $u$.

Consider the example in 
Fig.~\ref{fig:admin changes policy}. 
The system remains insecure in the intransitive setting:
as $A$ must not interfere with any agent in any state, 
we have $\dom(a)\notin\dsrc{ah}L{\epsilon}$, where again, according to our convention, $\epsilon$ denotes the initial state.
So, the system is insecure, since $\obs_L(ah)\neq\obs_L(h)$.

\subsection{Characterizations and Complexity of \dipsecuritysection}\label{sect:intransitive unwinding exponential}\label{sect:dipsecty characterizations}

We now establish two characterizations of intransitive noninterference with local policies and study the complexity of verifying \dipsecurity. Our characterizations are analogous to the ones obtained for the transitive setting in Theorem~\ref{thm:dp_characterizations}. The first one is based on a purge function, the second one uses an unwinding condition. This demonstrates the robustness of our definition and strengthens our belief that \dipsecurity is indeed a natural notion. 

We first extend Rushby's definition of $\ipurge$ to systems with local policies. 

\begin{definition}[intransitive purge for local policies]
  For all $u \in \Dom$ and all $s \in \States$, let $\dipurge{\epsilon}us=\epsilon$
  and, for all 
  $a \in \Actions$ and $\alpha \in \Actions^*$, let
\begin{align*} 
\dipurge{a\alpha}{u}{s} & =
 \begin{cases}
  a \beforeOpSpace \dipurge{\alpha}u{s\cdot a} & \mathtext{ if }\dom(a)\in\dsrc{a\alpha}{u}{s}, \\
  \dipurge{\alpha}u{s} & \mathtext{ otherwise}.
 \end{cases}
\end{align*}
 \end{definition}

The crucial point is that in the case where $a$ must remain hidden from
agent~$u$, we define $\dipurge{a\alpha}us$ as $\dipurge\alpha us$ instead
of the possibly more intuitive choice $\dipurge\alpha u{s\cdot a}$, on which 
the security definition in~\cite{Leslie-DYNAMICNONINTERFERENCE-SSE-2006} is
based.

We briefly explain the reasoning behind this choice. To this end, let $\dipurgename'$ denote the alternative definition of \dipurgename\xspace outlined above. Consider the sequence $ah$, performed from the initial state in the system in Fig.~\ref{fig:admin changes policy}. Clearly, the action~$a$ is purged from the trace, thus the result of $\dipurgename'$ is the same as applying $\dipurgename'$ to the sequence $h$ starting in the upper left state. However, in this state, the action $h$ is invisible for $L$, hence $\dipurgename'$ removes it, and thus purging $ah$ results in the empty sequence. On the other hand, if we consider the sequence $h$ also starting in the initial state, then $h$ is not removed by $\dipurgename'$, since $H$ may interfere with $L$. Hence $ah$ and $h$ do not lead to the same purged trace---a security definition based on $\dipurgename'$ does not require $ah$ and $h$ to lead to states with the same observation. Therefore, the system is considered secure in the $\dipurgename'$-based security definition from~\cite{Leslie-DYNAMICNONINTERFERENCE-SSE-2006}.  However, a natural definition must require $ah$ and $h$ to lead to the same observation for agent $L$, as the action $a$ must always be hidden from $L$. 

We next define unwindings for \dipsecurity and then give a characterization of \dipsecurity based on them.

\begin{definition}[intransitive unwinding]
An \emph{intransitive unwinding} for the system $M$ with a policy $\dpol$ is a family of relations $(\direl{\Dom'})_{\Dom'\subseteq \Dom}$ such that $\direl{\Dom'}\subseteq \States \times \States$ and for all $\Dom'\subseteq \Dom$, all $s, t \in \States$ and all $a \in \Actions$, the following hold:
\begin{itemize}
 \item $s\direl{\set{u\in \Dom\ \vert\ \dom(a)\not\dintrel{s} u}} s\cdot a$.
  \hfill \gLR
 \item  If $s\direl{\Dom''} t$, then $s\cdot
   b\direl{\Dom''} t\cdot b$, where $\Dom''= \Dom'$ if
   $\dom(b)\in\Dom'$, 
\\
and
   else $\Dom''= \Dom'\cap\set{u\ \vert\ \dom(b)\not\dintrel{s} u}$.
    \hfill \gSC
 \item If $s\direl{\Dom'} t$ and $u\in\Dom'$, then $\obs_u(s)=\obs_u(t)$,
  \hfill \gOC
\end{itemize}
\end{definition}

Intuitively, $s\direl{\Dom'} t$ expresses that there is a common reason for
all agents in~$\Dom'$ to have the same observations in $s$ as in $t$,
i.e., if there is a state $\tilde s$, an action~$a$ and a sequence
$\alpha$ such that $s= \tilde s\cdot a\alpha$, $t=\tilde s
\cdot\alpha$, and $\dom(a)\notin\dsrc{a\alpha}u{\tilde s}$ for
\emph{all} agents $u\in\Dom'$. 

\begin{theorem}[characterization of \dipsecurity]
\label{theorem:intransitive unwinding and ipurge characterization}
 The following are equivalent:
 \begin{enumerate}
  \item The system $M$ is \dipsecure. 
  \item For all agents $u$, all states $s$, and all action sequences
    $\alpha$ and $\beta$ with 
    \\
$\dipurge\alpha us=\dipurge\beta us$, we
    have $\obs_u(s\cdot\alpha)=\obs_u(s\cdot\beta)$. 
  \item There exists an intransitive unwinding for $M$ and $\dpol$.
 \end{enumerate}
\end{theorem}

In contrast to the transitive setting, the unwinding characterization of
\dipsecurity does not lead to a polynomial-time algorithm to verify
security of a system, because the number of relations needed to consider is exponential in the number of agents in the system. Unless $\PTIME=\NP$, we cannot do significantly better, because the verification problem is \NP-complete; our unwinding characterization, however, yields an FPT-algorithm.

\begin{theorem}[complexity of \dipsecurity]
  \label{theorem:intransitive case np complete} 
  Deciding whether a given system is \dipsecure with respect to a policy is \NP-complete and fixed-parameter tractable with the number of agents as parameter.
\end{theorem}

\subsection{Intransitively Useless Edges}

\begin{wrapfigure}[11]{r}{0pt}
\scalebox{0.75}{
 \begin{tikzpicture}[tikzglobal]  
    \newcommand{\agentshdl}{ 
      \node[agent] (h) [left] {$H$};
      \node[agent] (d) [right=2.5mm of h] {$D$};
      \node[agent] (l) [below=2.5mm of d] {$L$}; 
    }

    \newcommand{\agentshd}{ 
      \node[agent] (h) {$H$};
      \node[agent] (d) [below=2.5mm of h] {$D$};
    }

    \newcommand{\agentsdl}{ 
      \node[agent] (d) {$D$};
      \node[agent] (l) [below=2.5mm of d] {$L$}; 
    }

 \node[initial,systemstate] (q0) {
   \agentshd
   \\
   \hline
   $\obs_L \colon 0$ \\
  };
   \path[policy] (h) edge (d);

 \node[systemstate] (q1) [right= 10mm of q0] {
 \agentshdl
 \\ \hline
   $\obs_L \colon 0$ \\
  };
  \path[policy] (h) edge (l);
  \path[policy] (d) edge (l);

 \node[systemstate] (q3) [right=10mm of q1] {
  $\obs_L \colon 1$ \\
  };

 \node[systemstate] (q2) [below=5mm of q1] {
  \agentsdl  \\ 
  \hline
    $\obs_L \colon 0$ \\
  };
\path[policy] (d) edge (l);

 \node[systemstate] (q4) [right=10mm of q2] {
  $\obs_L \colon 2$ \\
  };

 \node[systemstate] (q5) [above=of q3] {
  $\obs_L \colon 0$ \\ };

 \node[systemstate] (q6) [left=20mm of q5] {
  $\obs_L \colon 0$ \\ };

 \path (q0) edge node {$h_1$} (q1);
 \path (q0) edge [bend right]  node {$h_2$} (q2);
 \path (q1) edge node {$d$} (q3);
 \path (q2) edge node {$d$} (q4);
 \path (q1) edge node {$h_1$} (q5);
 \path (q1) edge node {$h_2$} (q6);
\end{tikzpicture}}
\caption{Intransitively useless edge}\label{fig:redundant edge}
\end{wrapfigure}
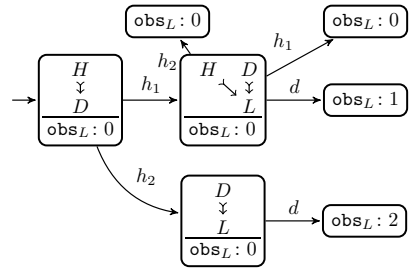

In our discussion of \dpsecty we observed that local policies may contain edges that can never be used. This issue also occurs in the intransitive setting, but the situation is more involved. In the transitive setting, it is sufficient to ``remove any incoming edge for $u$ that $u$ must not know about''  (see Theorem~\ref{theorem:uniformpolicies}). In the intransitive setting it is not: when the system in Fig.~\ref{fig:redundant edge} is in state $h_1$, then agent $L$ must not know that the edge $D\dintrel{} L$ is present, since states $\epsilon$ and $h_1$ should be indistinguishable for $L$, but clearly, the edge cannot be removed without affecting security. However, useless edges still exist in the intransitive setting, even in the system from Figure~\ref{fig:redundant edge}, as we will show below.

To formally define useless edges, we adapt \dpsimilarity to the intransitive setting in the natural way. 

\begin{definition}[\dipsimilarity]
  For an agent $u$, let $\approx^i_u$ be the smallest equivalence
  relation on the states of $M$ such that for all $s \in \States$, $a
  \in \Actions$, $\alpha\in \Actions^*$, if $\dom(a)\notin\dsrc{a\alpha}us$, then $s\cdot a\alpha\approx^i_us\cdot\alpha$. We call states $s$ and $s'$ with $s\approx^i_u s'$  \emph{\dipsimilar for $u$}.
\end{definition}

Using this, we can now define intransitively useless edges:

\begin{definition}[intransitively useless edge]
  Let $e$ be an edge in a local policy of \dpol and let $({\hat \rightarrowtail}_s)_{s\in S}$ be the policy obtained from \dpol by removing~$e$. Let $\approx^i_u$ and ${\hat \approx}^i_u$ be the respective i-similarity relations. Then $e$ is \emph{intransitively useless} if $s\approx^i_u s'$ if and only if $s {\hat \approx}^i_u s'$ for all states $s$ and $s'$ and all agents $u$.
\end{definition}

An edge is intransitively useless if removing it does not forbid any information flow that was previously allowed. In particular, such an edge itself cannot be used directly.
Whether an edge is useless does not depend on the observation function of the system, but only on the policy and the transition function, whereas a definition of security compares observations in different states.

If the policy does not contain any intransitively useless edges, then there is no edge in any of its local policies that is contradicted by other aspects of the policy. In other words, the set of information flows \emph{forbidden} by such  a policy is closed under logical deduction---every edge that can be shown to represent a forbidden information flow is absent in the policy. 

Fig.~\ref{fig:redundant edge} shows a secure system with an intransitively useless edge. The system is secure (agent $L$ knows whether in the initial state, $h_1$ or $h_2$ was performed, as soon as this information is transmitted by agent $D$). The edge $H\dintrel{h_1}L$ is intransitively useless, as explained in what follows. 

The edge allows $L$ to distinguish between the states $h_1, h_1h_1, h_1h_2$. However, one can verify that $h_2h_1\approx^i_L h_1$,  $h_2h_1h_1\approx^i_L h_2h_1$,  $h_2h_1h_1\approx^i_L h_1h_1$,  $h_2h_1h_2\approx^i_L h_2h_1$, and $h_2h_1h_2\approx^i_L h_1h_2$ all hold.
Symmetry and transitivity of $\approx^i_L$ imply that all the three states $h_1,h_1h_1,h_1h_2$ are $\approx^i_L$-equivalent. Hence the edge $H\dintrel{h_1}L$ is indeed intransitively useless (and the system would be insecure if $h_1$, $h_1h_1$, and $h_1h_2$ would not have the same observations).

Intransitively useless edges can be removed without affecting security: 

\begin{theorem}[removal of intransitively useless edges]
\label{theorem:redundant edges}
  Let $\dpolprime$ be obtained from $\dpol$ by removing a set of edges which are intransitively useless. Then $M$ is 
\dipsecure
with respect to $\dpol$ if and only if $M$ is 
\dipsecure
with respect to $\dpolprime$.
\end{theorem}

This theorem implies that for every policy $\dpol$, a policy $\dpolprime$ without intransitively useless edges that is equivalent to $\dpol$ can be obtained from $\dpol$ by removing all intransitively useless edges. 

\subsection{Sound Unwindings and Uniform Intransitive Policies}

The exponential size unwinding of \dipsecurity given in Section~\ref{sect:intransitive unwinding exponential} does not yield a polynomial-time algorithm for security
verification. Since the problem
is \NP-complete, such an algorithm---and hence an unwinding that is both small and easy to compute---does not exist, unless $\PTIME=\NP$. In this section, we define unwinding conditions that lead to a polynomial-size unwinding and are \emph{sound} for 
\dipsecurity, and are \emph{sound and complete} for 
\dipsecure in the case of uniform policies. Uniform policies are (as in the transitive case) policies in which every agent ``may know'' the set of agents who may currently interfere with him, that is, if an agent $u$ must not distinguish two states by the security definition, then the set of agents that may interfere with $u$ must be identical in these two states. 
Formally, we define this property as follows.

\begin{definition}[intransitive uniform]
  A policy $\dpol$ is \emph{intransitively uniform}, if for all agents $u$ and states $s$, $s'$ with $s\approx^i_u s'$, we have that $\infagents{u}{s}=\infagents{u}{s'}$.
\end{definition}

Note that this definition is very similar to the uniformity condition for the transitive setting, but while in the transitive setting, uniform policies and policies without useless edges coincide, this is not true for intransitive noninterference (in fact, neither implication holds). 

Uniformity, on an abstract level, is a natural requirement and often met in concrete systems, since an agent usually knows the sources of information available to him. In the uniform setting, many of the subtle issues with local policies do not occur anymore; as an example, \dipsecurity\ and the security definition from~\cite{Leslie-DYNAMICNONINTERFERENCE-SSE-2006} coincide for uniform policies. Uniformity also has nice algorithmic properties, as both, checking whether a system has a uniform policy and checking whether a system with a uniform policy satisfies \dipsecurity, can be performed in polynomial time. This follows from the characterizations of i-security in terms of the unwindings we define next.

\begin{definition}[uniform intransitive unwinding]
 A \emph{uniform intransitive unwinding} for $M$ with a policy \dpol\ is a family of equivalence relations $\sim^{\tilde s,v}_u$ for each
 choice of states $\tilde s$ and agents $v$ and $u$,
 such that for all $s, t \in \States$, and all $a \in \Actions$, the following holds:
 \begin{itemize}
  \item  If $s\sim^{\tilde s,v}_ut$, then $\obs_u(s)=\obs_u(t)$. 
    \hfill \dOC
  \item  If $s\sim^{\tilde s,v}_ut$, then $\infagents{u}{s}=\infagents{u}t$. 
    \hfill \dPC
  \item If $s\sim^{\tilde s,v}_ut$ and $a\in
    A$ with $v\not\dintrel{\tilde s}\dom(a)$, then $s\cdot
    a\sim^{\tilde s,v}_ut\cdot a$.
      \hfill \dSC
  \item If $\dom(a)\not\dintrel{\tilde s}u$,
    then $\tilde s\sim^{\tilde s,\dom(a)}_u \tilde s\cdot a$.
    \hfill \dLR
 \end{itemize}
\end{definition}

In the following theorem intransitive uniformity and \dipsecurity (for uniform policies) are characterized by almost exactly the same unwinding. The only difference is that for uniformity we require policy consistency \dPC, since we are concerned with having the same \emph{local policies} in certain states, while for security, we require \dOC, since we are interested in \emph{observations}.

\pagebreak
\begin{theorem}[uniform unwinding characterizations]
  \label{theorem:polynomial unwinding characterization of intransitive uniformity and security}
  \begin{enumerate}
  \item The policy $\dpol$ is intransitively uniform if and only if there is a uniform intransitive unwinding for $M$ and $\dpol$ that satisfies \dPC, \dSC, and \dLR.
  \item If $\dpol$ is intransitively uniform, then $M$ is \dipsecure if and only if there is a uniform intransitive unwinding  that satisfies \dOC, \dSC and \dLR.
  \end{enumerate}
\end{theorem}

In particular, if an unwinding satisfying all four conditions exists, then a system is secure.
Due to Theorem~\ref{theorem:intransitive case np complete}, we cannot hope that the above unwindings completely characterize \dipsecurity, and indeed the system in Fig.~\ref{fig:redundant edge} is 
\dipsecure but not intransitively uniform.
However, for uniform policies, Theorem~\ref{theorem:polynomial unwinding characterization of intransitive uniformity and security} immediately yields efficient algorithms to verify the respective conditions via a standard dynamic programming approach:

\begin{corollary}[uniform unwinding verification]
\begin{enumerate}
 \item Verifying whether a policy is intransitively uniform can be performed in nondeterministic logarithmic space.
 \item For systems with intransitively uniform policies, verifying whether a system is 
   \dipsecure
can be performed in nondeterministic logarithmic space.
\end{enumerate}
 \end{corollary}

The above shows that the complexity of intransitive noninterference with local policies comes from the \emph{combination} of local policies that do not allow agents to ``see'' their allowed sources of information with an intransitive security definition. In the transitive setting, this interplay does not arise, since there a system always can allow agents to ``see'' their incoming edges (see Theorem~\ref{theorem:uniformpolicies}).

\subsection{Unwinding for \ipsecuritysection}

In the setting with a global policy, \dipsecurity
is equivalent to IP-security as defined in~\cite{HY87}. For
IP-security, Rushby gave unwinding conditions that are sufficient, but
not necessary. 
This left open the question whether there is an unwinding condition that \emph{exactly} characterizes IP-security, which we can now answer positively as follows. Clearly, a policy that assigns the same local policy to every state is intransitively uniform. Hence our results immediately yield a characterization of IP-security with the above unwinding conditions, and from these, an algorithm verifying IP-security in nondeterministic logarithmic space can be obtained in the straight-forward manner.

\begin{corollary}[unwinding for \ipsecurity]
 \begin{enumerate}
  \item A system is IP-secure if and only if it has an intransitive unwinding satisfying \dOC, \dSC, and \dLR.
  \item IP-security can be verified in nondeterministic logarithmic space.
 \end{enumerate}
\end{corollary}

\section{Conclusion}

We have shown that noninterference with local policies is considerably different from noninterference with a global policy: an allowed interference in one state may contradict a forbidden interference in another state. Our new definitions address this issue. Our purge- and unwinding-based characterizations show that our definitions are natural, and directly lead to our complexity results. 

We have studied generalizations of Rusby's IP-security~\cite{rushby92}. An interesting question is to study van der Meyden's TA-security~\cite{meyden2007} in a setting with local policies. Preliminary results indicate that such a generalization needs to use a very different approach from the one used in this paper.

\bibliographystyle{alpha}

\section{Additional Results}

In this Section we present and prove additional results which were informally mentioned in the main paper. 
 
\subsection{Initial-State Verification Suffices for Uniform Policies}

One noteworthy difference to the case of a system with a global policy is that it is
necessary to evaluate the $\dpurge$-function in every state,
and not only in the initial state:
The system in Figure~\ref{fig:dpurge_in_all_states}
is secure with respect to the \dpurge-based characterization of \dpsecty, if we only consider traces starting in the initial
state, but can easily be seen to not be \dpsecure. 

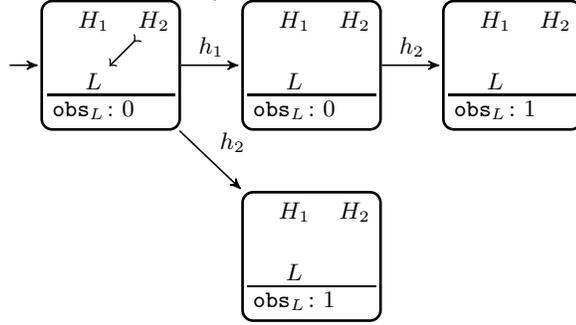
\begin{figure}[h]
 \begin{tikzpicture}[tikzglobal]
   
    \newcommand{\agents}{ 
      \node[agent] (h1) {$H_1$};
      \node[agent] (h2) [right of=h1] {$H_2$};
      \node[agent] (l) [below of=h1] {$L$}; 
}

 \node[initial,systemstate] (q0) {
   \agents
   \\
   \hline
   $\obs_L \colon 0$ \\
};
   \path[policy] (h2) edge node {} (l);
   
 \node[systemstate] (q1) [right=of q0] {
   \agents
  \\
   \hline
   $\obs_L \colon 0$ \\
};
 \node[systemstate] (q2) [right=of q1] {
   \agents
  \\
   \hline
   $\obs_L \colon 1$ \\
};
\node[systemstate] (q3) [below=of q1] {
   \agents
  \\
   \hline
   $\obs_L \colon 1$ \\
};

\path (q0) edge node {$h_1$} (q1) 
edge node {$h_2$} (q3)
(q1) edge node {$h_2$} (q2)
;
\end{tikzpicture}
\caption{System with a non-uniform policy}
\label{fig:dpurge_in_all_states}
\end{figure}

However, in the case of a uniform policy, it suffices to consider traces starting in the initial state, as we now show.

\begin{theorem}
  \label{theorem:dpurge_uniform_security}
  Let $M$ be a system with a uniform policy. 
  Then $M$ is 
  \dpsecure
 iff for all $u \in \Dom$ and all $\alpha \in \Actions^*$: 
  $\obs_u(s_0 \cdot \alpha) = \obs_u(s_0 \cdot \dpurge(\alpha, u, s_0))$. 
\end{theorem}

\begin{proof}
  Assume that $M$ is a secure system. 
  Then from $s_0 \cdot \alpha \unwind_u s_0 \cdot \dpurge(\alpha, u,
  s_0)$ follows from the output consistency that
  $\obs_u(s_0 \cdot \alpha) = \obs_u(s_0 \cdot \dpurge(\alpha, u,
  s_0))$. 
  
  For the other direction of the proof, we consider $\alpha, \beta
  \in\Actions^*$ with $\dpurge(\alpha, u, s)$ $= \dpurge(\beta, u, s)$. 
  Then it exists $\gamma \in \Actions^*$ with $s = s_0 \cdot \gamma$. 
  It follows that $s_0 \cdot \gamma \unwind_u \dpurge(\gamma, u,
  s_0)$. 
  This gives
  \begin{align*}
    \obs_u(s \cdot \alpha) 
    & = \obs_u(s_0 \cdot \gamma \alpha) \\
    & = \obs_u(s_0 \cdot \dpurge(\gamma \alpha, u, s_0)) \\
    & = \obs_u(s_0 \cdot \dpurge(\gamma, u, s_0) \dpurge(\alpha, u,
    s_0 \cdot \dpurge(\gamma, u, s_0))) \\
    & = \obs_u(s_0 \cdot \dpurge(\gamma,u, s_0) \dpurge(\alpha, u, s_0
    \cdot \gamma)) \\
    & = \obs_u(s_0 \cdot \dpurge(\gamma,u, s_0) \dpurge(\beta, u, s_0
    \cdot \gamma)) \\
    & = \obs_u(s_0 \cdot \beta)
    \enspace. 
  \end{align*}
\end{proof}

\subsection{Some Properties of the \dpurge\  Function}

Here we show that our \dpurge\ function in the transitive setting behaves very naturally in the case of a uniform policy.

\begin{lemma}
  \label{lemma:dpurge_properties}
  Let $M$ be a system with a policy $\dpol$.
  For every $u \in \Dom$, $s, t \in \States$ and $\alpha, \beta \in \Actions^*$,
  we have
  \begin{enumerate}
  \item $\dpurge(\dpurge(\alpha, u, s), u, s) = \dpurge(\alpha, u,
    s)$,
  \item $\dpurge(\alpha \beta, u, s) = \dpurge(\alpha, u, s)
    \dpurge(\beta, u, s \cdot \dpurge(\alpha, u, s))$,
  \item if $\dpol$ is uniform and if $\unwind_u$ is an equivalence relation on $\States$
    that satisfies \lrdp and \scdp and if $s \unwind_u t$, then 
    $s \cdot \alpha \unwind_u t \cdot \dpurge(\alpha, u, t)$ and 
    $\dpurge(\alpha, u, s) = \dpurge(\alpha, u, t)$. 
  \end{enumerate}
\end{lemma}

\begin{proof}
  \begin{enumerate}
  \item We show this by an induction on the length of $\alpha$. 
    Since the base case is obvious, we proceed with the inductive
    step. 
    We consider $a \alpha$ with $a \in \Actions$ and $\alpha \in
    \Actions^*$ and assume that the claim holds for $\alpha$. 
    In the following two cases, we get
    \begin{enumerate}
    \item If $\dom(a) \dintrel{s} u$,
      we have
      \begin{align*}
        \dpurge(\dpurge(a \alpha, u, s), u, s) 
        & = \dpurge(a \dpurge(\alpha, u, s\cdot a), u s)\\
        & = a \dpurge(\dpurge(\alpha, u, s\cdot a), u, s\cdot a) \\
        & \stackrel{\text{I.H.}}{=} a \dpurge(\alpha, u, s\cdot a)\\
        & = \dpurge(a \alpha, u, s)
          \enspace.
        \end{align*}
      \item If $\dom(a)\not\dintrel{s} u$, 
        we have
        \begin{align*}
          \dpurge(\dpurge(a \alpha, u, s), u, s)
          & = \dpurge(\dpurge(\alpha, u, s), u, s) \\
          &  \stackrel{\text{I.H.}}{=} \dpurge(\alpha, u, s)
          \enspace.
        \end{align*}
    \end{enumerate}
  \item We show this claim by an induction on the length of $\alpha$
    and consider again $a \alpha$. 
    We get the following two cases
    \begin{enumerate}
    \item If $\dom(a) \dintrel{s} u$, 
      we have
      \begin{align*}
        \dpurge(a \alpha \beta, u, s) 
        & = a \dpurge(\alpha \beta, u, s \cdot a) \\
        & \stackrel{\text{I.H.}}{=} a \dpurge(\alpha, u, s\cdot a)
        \dpurge(\beta, u, s \cdot a \dpurge(\alpha, u, s\cdot a)) \\
        & = \dpurge(a \alpha, u, s) \dpurge(\beta, u, s \cdot
        \dpurge(a \alpha, u, s)) 
        \enspace. 
      \end{align*}
    \item If $\dom(a) \not\dintrel{s} u$, 
      we have
      \begin{align*}
        \dpurge(a \alpha \beta, u, s) 
        & = \dpurge(\alpha \beta, u, s) \\
        & \stackrel{\text{I.H.}}{=} \dpurge(\alpha, u,
        s)\dpurge(\beta, u, s \cdot \dpurge(\alpha, u, s)) \\
        & = \dpurge(a \alpha, u, s) \dpurge(\beta, u, s \cdot \dpurge(a
        \alpha, u, s))
        \enspace. 
      \end{align*}
    \end{enumerate}
  \item 
    This can be shown by an induction on the length of $\alpha$. 
  \end{enumerate}
\end{proof}

\subsection{Equivalence of Intransitive Security Definitions for Uniform Policies}

We now show that in case of an intransitively uniform policy, a system is secure with respect to the definition of~\cite{Leslie-DYNAMICNONINTERFERENCE-SSE-2006} if and only if it is \dipsecure.

We first show the following Lemma, which intuitively says that if the first action of $a\alpha$ is not transmitted to $u$ on the path $a\alpha$, then the same actions on the remaining path $\alpha$ are transmitted to $u$ when evaluating $\alpha$ from the state $s$ or from the state $s\cdot a$ in the case of a uniform policy. This is the key reason why, for uniform policies, the difference between Leslie's function $\dipurgeprimename$ and our $\dipurgename$ is irrelevant.

\begin{lemma}\label{lemma:same dsources on paths in uniform policies}
 Let $M$ be a system with an intransitively uniform policy $\dpol$. Let $\dom(a)\notin\dsrc{a\alpha}us$, where $\alpha=\beta b\beta'$. Then 
$$\dom(b)\in\dsrc{b\beta'}u{s\cdot\beta}\mathtext{ iff
}\dom(b)\in\dsrc{b\beta'}{u}{s\cdot a \beta}.$$
\end{lemma}

\begin{proof}
 Assume this is not the case, and let $b\beta'$ be a minimal
 counter-example. First assume that
 $\dom(b)\in\dsrc{b\beta'}{u}{s\cdot a\beta}$ and
 $\dom(b)\notin\dsrc{b\beta'}u{s\cdot \beta}$. Then there is some
 $\dom(c)\in\dsrc{\beta'}u{s\cdot a\beta  b}$ with $\dom(b)\dintrel{s\cdot a\beta}\dom(c)$, and due to
 minimality of $b\beta'$ it follows that $\dom(c)\in\dsrc{\beta'}u{s\cdot \beta b}$. Since $\dom(b)\notin\dsrc{b\beta'}u{s\cdot\beta}$, it thus follows that $\dom(b)\not\dintrel{s\cdot\beta}\dom(c)$. This is a contradiction to the intransitive uniformity of $\dpol$, since $\dom(a)\notin\dsrc{a\beta}{\dom(c)}s$, and hence $s\cdot a\beta\approx^i_{\dom(c)}s\cdot\beta$.

 The second case is essentially identical: Assume that $\dom(b)\in\dsrc{b\beta'}u{s\cdot\beta}$ and $\dom(b)\notin\dsrc{b\beta'}{u}{s\cdot a\beta}$. Then there is some $\dom(c)\in\dsrc{\beta'}u{s\cdot\beta b}$ with $\dom(b)\dintrel{s\cdot\beta}\dom(c)$. Due to the minimality of $b\beta'$, it follows that $\dom(c)\in\dsrc{\beta'}u{s\cdot a\beta b}$, hence $\dom(b)\not\dintrel{s\cdot a\beta}\dom(c)$. Since $s\cdot a\beta\approx^i_{\dom(c)}s\cdot\beta$ due to the above, we have a contradiction to the uniformity of $\dpol$.
\end{proof}

From the above Lemma, we can now easily show that for uniform policies, \dipsecurity\ and security in the sense of~\cite{Leslie-DYNAMICNONINTERFERENCE-SSE-2006} coincide:

\begin{theorem}
 Let $M$ be a system with an intransitively uniform policy $\dpol$. Then $M$ is \dipsecure\ if and only if $M$ is secure with respect to the definition in~\cite{Leslie-DYNAMICNONINTERFERENCE-SSE-2006}.
\end{theorem}

\begin{proof}
 Due to Theorem~\ref{theorem:intransitive unwinding and ipurge characterization}, it suffices to show that in the case of a uniform policy, the functions $\dipurgename$ and $\dipurgeprimename$ coincide. Assume indirectly that this is not the case, and let $\alpha$ be a minimal sequence such that there exists a state $s$ and an agent $u$ with $\dipurge{\alpha}us\neq\dipurgeprime{\alpha}us$. Clearly $\alpha\neq\epsilon$, hence assume that $\alpha=a\alpha'$.

 First assume that $\dom(a)\in\dsrc{a\alpha'}us$. In this case, we have (by definition and minimality of $\alpha$), that

$$
 \begin{array}{llllllllllllll}
  \dipurge{a\alpha'}us & = & a \beforeOpSpace \dipurge{\alpha'}u{s\cdot a}  \\
                       & = & \dipurgeprime{\alpha'}{u}{s\cdot a} & = &\dipurge{a\alpha'}{u}{s} \enspace,
 \end{array}
$$

 which is a contradiction to the choice of $\alpha$.

 Hence assume that $\dom(a)\notin\dsrc{a\alpha'}us$. By definition, it follows that $\dipurge{a\alpha'}us=\dipurge{\alpha'}us$ and $\dipurgeprime{a\alpha'}{u}{s}=\dipurgeprime{\alpha'}{u}{s\cdot a}=\dipurge{\alpha'}u{s\cdot a}$ (the final equaility is due to the minimality of $\alpha$).

 It hence suffices to show that $\dipurge{\alpha'}us=\dipurge{\alpha'}u{s\cdot a}$. This easily follows by induction on Lemma~\ref{lemma:same dsources on paths in uniform policies}: The same actions of $\alpha'$ are transmitted to $u$ when evaluating $\alpha'$ starting in the state $s$ and in $s\cdot a$.
\end{proof}

\section{Proofs}

In this section we give proofs for the results claimed in the paper.

\subsection{Proof of Theorem~\ref{thm:dp_characterizations}}
\begin{proof}
  First, we will show that  \ref{thm:dp_characterization_def}. implies \ref{thm:dp_characterization_unwind}.. 
  Let $M$ be a \dpsecure system. 
  Let $u \in \Dom$. Define for every $s, t \in \States$:
  \begin{equation*}
    s \unwind_u t \text{ iff for all } \alpha \in \Actions^* : 
    \obs_u(s \cdot \alpha) = \obs_u(t \cdot \alpha) \enspace. 
  \end{equation*}
  The condition \ocdp is satisfied if $\alpha = \epsilon$. 
  For the condition \scdp, we consider $s, t \in \States$ with $s \unwind_u t$ and let $a \in \Actions$. 
  Then for all $\alpha \in \Actions^*$, we have $s \cdot \alpha
  \unwind_u t \cdot\alpha$ and also $s\cdot a \alpha \unwind_u t \cdot
  a \alpha$. Therefore, $s \cdot a \unwind_u t\cdot a$.
  For the condition \lrdp, we consider $a \in \Actions$ and $s \in \States$ with $\dom(a) \not\dintrel{s} u$. Since $s$ is a reachable state, it exists $\alpha \in \Actions^*$ with $s = s_0 \cdot \alpha$. The definition of \dpsecurity states, that for every $\beta \in \Actions^*$ the equality of $\obs_u(s\cdot a \beta)$ and $\obs_u(s \cdot \beta)$ holds. 
 Therefore, $s \unwind_u s\cdot a$. 

 We assume that \ref{thm:dp_characterization_unwind}. holds and 
 will proof \ref{thm:dp_characterization_purge}.. 
 Let $u\in \Dom$ and assume that there exists a transitive unwinding $\unwind_u$ that satisfies \lrdp,
  \scdp and \ocdp. 
  We will show by an induction on the combined length of $\alpha$ and
  $\beta$, that for every state $s \in \States$: 
  $\dpurge(\alpha, u, s) = \dpurge(\beta, u, s)$ implies $s \cdot
  \alpha \unwind_u s \cdot \beta$. 
  The base case with $\alpha = \beta = \epsilon$ is clear. 
  For the inductive step consider $\alpha$ and $\beta$ with
  $\dpurge(\alpha, u, s) = \dpurge(\beta, u, s)$ for some state $s$. 
  We have to consider two cases:
  \begin{enumerate}[{Case }1:]
  \item 
    $\alpha = a \alpha'$ for some $a \in \Actions$, $\alpha' \in
    \Actions^*$ and $\dom(a) \not\dintrel{s} u$. 
    Then we have $\dpurge(a \alpha', u ,s) = \dpurge(\alpha', u, s)$.
    From the property \lrdp  follows that $s \unwind_u
    s \cdot a$ and from \lrdp follows $s \cdot \alpha' \unwind_u s
    \cdot a \alpha'$. 
    Applying the induction hypothesis gives $s\cdot \alpha' \unwind_u
    s \cdot \beta$ which can be combined to $s \cdot \alpha \unwind_u
    s \cdot \beta$.
  \item 
    $\alpha = a \alpha'$ and $\beta = b \beta'$ with 
    $\dom(a) \dintrel{s} u$ and $\dom(b) \dintrel{s} u$. 
    From 
    \begin{align*}
      a \beforeOpSpace \dpurge(\alpha', u, s\cdot a) 
      & = \dpurge(a \alpha', u, s) \\
      & = \dpurge(\alpha, u, s) \\
      & = \dpurge(\beta, u, s) \\
      & = b \beforeOpSpace \dpurge(\beta', u, s \cdot b)
    \end{align*}
    follows that $a = b$ and $\dpurge(\alpha', u, s\cdot a) =
    \dpurge(\beta', u, s\cdot a)$. 
    Applying the induction hypothesis gives
    $s \cdot a \alpha' \unwind_u s \cdot b \beta'$. 
  \end{enumerate}
  In both cases follows from \ocdp that $\obs_u(s \cdot \alpha) =
  \obs_u(s \cdot \beta)$. 

  For proofing the implication
  from~\ref{thm:dp_characterization_purge}
  to~\ref{thm:dp_characterization_def}, we assume, that $M$ does
  not satisfy \dpsecurity.
  Therefore, there exists an agent $u \in \Dom$ and states $s, s'\in
  \States$ with $s \approx_u s'$ and $\obs_u(s) \neq \obs_u(s')$. 
  By the definition of \dpsecurity, there exists $t \in \States$, $a
  \in \Actions$ and $\alpha \in \Actions^*$ with $\dom(a)
  \not\dintrel{t} u$, $s = t \cdot a \alpha$ and $s' = t \cdot
  \alpha$. 
  By applying of $\dpurge$, we have $\dpurge(a \alpha, u, t) =
  \dpurge(\alpha, u, t)$ and from $\obs_u(t \cdot a \alpha) \neq
  \obs_u(t \cdot \alpha)$, follows
  that~\ref{thm:dp_characterization_purge} does not hold. 

For proofing the missing implication, we assume that  \ref{thm:dp_characterization_def}. does not hold. 
Therefore, it exists
$u \in \Dom$, 
$s \in \States$, 
$a \in \Actions$ and
$\alpha \in \Actions^*$ 
with $\dom(a) \not\dintrel{s} u$ and 
$\obs_u(s \cdot a \alpha) \neq \obs_u(s \cdot \alpha)$. 
Therefore, $s\cdot a \alpha \approx_u s \cdot \alpha$
and~\ref{thm:dp_characterization_def} does not hold.
\end{proof}

\subsection{Proof of Theorem~\ref{theorem:uniformpolicies}}
\begin{proof}
Let $M$ be a \dpsecure system with respect to the policy \dpol. 
Then there exists  a transitive unwinding $(\unwind_u)_{u \in
  \Dom}$  for $M$. 
Note, that for every $u \in \Dom$, the smallest eqivalence relation
$\unwind_u$ that satisfies \lrdp and \scdp is equal to the
smallest equivalence relation on $\States$ that includes $\approx_u$.
Let $\unwind_u'$ be the a smallest equivalence relation that satisfies \scdp and \lrdp with respect to the policy \dpolprime. 
We will show that ${\unwind_u'} \subseteq {\unwind_u}$.
Let $s, t \in \States$ with $s \unwind_u' t$ and $t = s \cdot a$ form some $a \in \Actions$ with $\dom(a) \not\dintrel{s}' u$. 
Therefore, there exists $s' \in \States$ with $s' \unwind_u s$ and $\dom(a) \not\dintrel{s'} u$. 
From $s' \unwind_u s' \cdot a$ and $s' \cdot a \unwind_u s \cdot a$ follows $s \unwind_u t$. 

The other direction of the proof follows directly from the fact, that the policy \dpolprime is at least as restrictive as the policy \dpol. 
\end{proof}

\subsection{Proof of Theorem~\ref{theorem:intransitive unwinding and ipurge characterization}}

\begin{proof}
  We first consider the \dipurgename-characterization and then the intransitive unwinding characterization.
 \begin{enumerate}
  \item 
 We first show that \dipsecurity\ implies the \dipurgename-characterization. 
 Hence indirectly assume that the system is \dipsecure, and indirectly assume that the \dipurgename-condition is not satisfied. 
Then there exists a state $s$, an agent $u$, and sequences $\alpha$ and $\beta$ with $\dipurge\alpha us=\dipurge\beta us$, and $\obs_u(s\cdot\alpha)\neq\obs_u(s\cdot\beta)$. We choose $\alpha$ and $\beta$ such that $\card\alpha+\card\beta$ is minimal among all such examples. Clearly, if \emph{both} $\alpha$ and $\beta$ start with an action that is transmitted to $u$, then this action must be the same: If $\alpha=a\alpha'$ with $\dom(a)\in\dsrc{a\alpha'}us$ and $\beta=b\beta'$ with $\dom(b)\in\dsrc{b\beta'}us$, then $\dipurge\alpha us$ starts with $a$, and $\dipurge\beta us$ starts with $b$. It thus follows that $a=b$, and hence we could use the state $s'=s\cdot a$ and the sequences $\alpha'$ and $\beta'$ as a counter-example, which contradicts the minimality of $\alpha$ and $\beta$. Hence we can, without loss of generality, assume that $\alpha=a\alpha'$ for some $a$ with $\dom(a)\notin\dsrc{a\alpha'}us$. It thus follows that 
$\dipurge{\alpha'}us=\dipurge\alpha us=\dipurge\beta us$. Since the system is secure, we also have $\obs_u(s\cdot\alpha')=\obs_u(s\cdot a\alpha')=\obs_u(s\cdot\alpha)\neq\obs_u(s\cdot\beta)$, and hence we again obtain a contradiction to the minimality of $\alpha$ and $\beta$ (with choosing $\alpha'$ instead of $\alpha$).

 We now show the converse, i.e., that the \dipurgename-characterization implies \dipsecurity. Hence assume that the system satisfies the \dipurgename-condition. To show interference security, let $\dom(a)\notin\dsrc{a\alpha}us$ for some agent $u$ and state $s$, we show that $\obs_u(s\cdot a\alpha)=\obs_u(s\cdot\alpha)$. Note that since $\dom(a)\notin\dsrc{a\alpha}us$, it follows that $\dipurge{a\alpha}us=\dipurge\alpha us$. Hence from the prerequisites of the theorem it follows that $\obs_u(s\cdot a\alpha)=\obs_u(s\cdot\alpha)$ as required.

 \item
 We prove that the intransitive unwinding characterization is also equivalent to \dipsecurity. 
 First assume that there is an intransitive unwinding $(\direl{\Dom'})_{\Dom'\subseteq \Dom}$ for $M$ with respect to $\dpol$. 
 We show that the system is \dipsecure. 
 For this it suffices to show that if $\dom(a)\notin\dsrc{a\alpha}us$, then $s\cdot a\alpha\direl{\Dom'} s\cdot\alpha$ for some set $\Dom'$ with $u\in \Dom'$. 
For each prefix $\alpha'$ of $\alpha$, let $\Dom_{\alpha'}$ be defined as 
\begin{equation*}
  \Dom_{\alpha'}=\set{v\in D\ \vert\ \dom(a)\notin\dsrc{a\alpha'}vs} \enspace.
\end{equation*}
Clearly, if $\alpha'$ is a prefix of $\alpha''$, then $\Dom_{\alpha''}\subseteq \Dom_{\alpha'}$.
Since $u\in \Dom_\alpha$, it suffices to show that $s\cdot a\alpha'\direl{\Dom_{\alpha'}} s\cdot\alpha'$ for all prefixes $\alpha'$ of $\alpha$. 
We show the claim by induction. 
For $\alpha'=\epsilon$, the claim follows from \gLR, since $\dom(a)\not\dintrel su$. 
Hence assume that $\alpha'=\beta b$ for some sequence $\beta$ and action $b$. 
By induction, we have that $s\cdot a\beta\direl{\Dom_\beta}s\cdot\beta$, 
where $\Dom_\beta$ contains all agents $v$ with $\dom(a)\notin\dsrc{a\beta}vs$. 
Now let $u\in \Dom_{\alpha'}$, it then also follows that $u\in \Dom_{\beta}$. 
Let $\Dom'$ be defined as in the condition \gSC. 
Since the condition implies $s\cdot a\beta b\direl{\Dom'}s\cdot\beta b$, 
it suffices to show that $u\in \Dom'$. 
Clearly this is the case if $\dom(b)\in \Dom_\beta$, i.e., if $\Dom_\beta=\Dom'$. 
Hence assume this is not the case, by definition of $\Dom_\beta$ it then follows that $\dom(a)\in\dsrc{a\beta}{\dom(b)}s$. 
Since $\dom(a)\notin\dsrc{a\beta b}us$, this implies that $\dom(b)\not\dintrel{s\cdot a\beta} u$, hence $u\in \Dom'$ follows in this case as well.

For the other direction, assume that the system is \dipsecure. 
We define $s\direl{\Dom'} t$ if there is a state $\tilde s$, an action $a$ and a sequence $\alpha$, such that $s=\tilde s \cdot a\alpha$, $t=\tilde s\cdot\alpha$, and for all $u\in \Dom'$, we have $\dom(a)\notin\dsrc{a\alpha}u{\tilde s}$. 
We claim that this defines an intransitive unwinding for $M$ with respect to $\dpol$.  
Since the system is \dipsecure, the condition \gOC is obviously satisfied. 
The condition \gLR follows from the fact that if $\dom(a)\not\dintrel su$, then $\dom(a)\notin\dsrc aus$. 
It remains to show \gSC. 
Hence let $s \direl{\Dom'} t$, and let $\tilde s$, $a$ and $\alpha$ be chosen with the above properties. 
Let $b$ be an action, and let $\Dom''$ be the set resulting from applying \gSC. 
It remains to show that for each $u\in \Dom''$, we have $\dom(a)\notin\dsrc{a\alpha b}u{\tilde s}$. 
First assume that $\dom(b)\in \Dom'$, it then follows from the definition of $\direl{\Dom'}$ that $\dom(a)\notin\dsrc{a\alpha}{\dom(b)}{\tilde s}$, and hence $\dom(a)\notin\dsrc{a\alpha b}u{\tilde s}$. 
On the other hand, if $\dom(b)\notin \Dom'$, then from $u\in \Dom''$, we know that $\dom(b)\not\dintrel{\tilde s\cdot a\alpha}u$, and hence from $\dom(a)\notin\dsrc{a\alpha}u{\tilde s}$ (since $u\in \Dom'$) and $\dsrc{a\alpha b}u{\tilde s}=\dsrc{a\alpha}u{\tilde s}$, it follows that $\dom(a)\notin\dsrc{a\alpha b}u{\tilde s}$ as required.
\end{enumerate}
\end{proof}

\subsection{Proof of Theorem~\ref{theorem:intransitive case np complete}}
\label{sect:complexity}

\begin{theorem}
 Checking whether a system is not \dipsecure can be done in \NP.
\end{theorem}

\begin{proof}
 The algorithm simply guesses the corresponding values of $a$, $u$, $s$, and $\alpha$, and verifies that these satisfy $\obs_u(s\cdot a\alpha)\neq\obs_u(s\cdot\alpha)$ and $\dom(a)\notin\dsrc{a\alpha}us$ in the straight-forward way. To show that this gives an \NP-algorithm, it suffices to show that the length of $\alpha$ can be bounded polynomially in the size of the system. We show that if the system is insecure, then $\alpha$ can be chosen with $\card\alpha\leq\card \States^2$.

 To show this, let $\alpha$ be a path of minimal length satisfying the above. Let $F_s$ and $F_{s\cdot a}$ be the finite state machines obtained when starting the system in the states $s$ and $s\cdot a$, respectively, and let $F=F_s\times F_{s\cdot a}$, with initial state $(s,s\cdot a)$. Clearly, in $F$, we have $(s,s\cdot a)\cdot \alpha=(s\cdot\alpha, s\cdot a\alpha)$. If $\card\alpha\ge \card \States^2$, then $\alpha$ visits a state from $F$ twice, i.e., $\alpha$ contains a nontrivial loop. Such a loop can be removed from $\alpha$ without changing the states that are reached. Clearly, removing a loop does not add information flow, hence the thus-obtained $\alpha'$ also satisfies the prerequisites for $\alpha$, which is a contradiction to $\alpha$'s minimality.
\end{proof}

\begin{theorem}\label{theorem:np hardness}
 For every security definition that is at least as strict as in\-for\-ma\-tion-flow-security and at least as permissive as interference-security, the problem to determine whether a given system is insecure is \NP-hard under $\redlogm$-reductions.
\end{theorem}

\begin{wrapfigure}[20]{r}{6.25cm}
 \begin{tikzpicture}[->,>=stealth',shorten >=1pt,auto,node distance=1cm, semithick]
  \tikzset{round-boxed/.style={draw=black, line width=1pt, dash pattern=on 1pt off 0pt, inner sep=0mm, rectangle, rounded corners}};

  \node(incoming) at (0,0) { $\cdot$ };

  \node (h)      at (3,3) { $h$ };
  \node (uneq0)   at (3,2) { $u_{\neq0}$ } ;
  \path (h) edge (uneq0);
  \node (u-0-1) [round-boxed, fit = (h) (uneq0)] {} ;

  \path (incoming) edge node[above] { $u_{=0}$} (u-0-1);

  \node (h)      at (3,0.5) { $h$ };
  \node (uneq1)   at (3,-0.5) { $u_{\neq1}$ } ;
  \path (h) edge (uneq1);
  \node (u-1-1) [round-boxed, fit = (h) (uneq1)] {} ;

  \path (incoming) edge node[above] { $u_{=1}$} (u-1-1);

  \node (h)      at (3,-2) { $h$ };
  \node (uneq2)   at (3,-3) { $u_{\neq2}$ } ;
  \path (h) edge (uneq2);
  \node (u-2-1) [round-boxed, fit = (h) (uneq2)] {} ;

  \path (incoming) edge node[above] { $u_{=2}$} (u-2-1);

  \node(outgoing) at (6,0) { $\cdot$ };

  \path (u-0-1) edge node[above] { $h$ } (outgoing);
  \path (u-1-1) edge node[above] { $h$ } (outgoing);
  \path (u-2-1) edge node[above] { $h$ } (outgoing);
 \end{tikzpicture}
  \caption{System $\colorsys u$}\label{fig:colorsys}
\end{wrapfigure}

We reduce from the 3-colorability problem for graphs. Let a graph $G$ with vertices $u_1,\dots,u_n$ and edges $(v^1_1,v^2_1),\dots,(v^m_1,v^m_2)$ be given. We construct a system $M^G$ as follows:

\begin{itemize}
 \item for each vertex $u$, there is an agent $u$ with actions $u_{=0}$, $u_{=1}$, and $u_{=2}$, and there are agents $u_{\neq 0}$, $u_{\neq 1}$, $u_{\neq 2}$, each having exactly one action, which for simplicity we denote with the agent's name. Additionally, there is an agent $h$ with a single action $h$, and an agent $L$ with a single action $L$.
 \item for each vertex $u$, we construct a subsystem $\colorsys u$ (see Figure~\ref{fig:colorsys}), that models the choice of coloring of $u$ in the graph. In $\colorsys u$ and all following systems, all transitions that are not explicitly indicated in the graphical representation loop in the corresponding state.
\end{itemize}

\begin{itemize}
 \item for each edge $(u,v)$, we construct a subsystem $\diffcols uv$ (see Figure~\ref{fig:diffcols}), which enforces that the colors of $u$ and $v$ must be different. The edges labelled with a transition of the form $u_{\neq i,j}$ represent two consecutive edges, the first one with the transition $u_{\neq i}$, and the second one labelled with the transition $u_{\neq j}$, where the policy is repeated between the two transitions.
 \item the system $M^G$ is now designed as shown in Figure~\ref{fig:complete system MG}. 
 We denote the left-most state with $s_0$. The unlabelled arrows between the different $\colorsys u$ and $\diffcols uv$-nodes express that the final node of one is the starting node of the other. The subsystems $\colorsysprime u$ and $\diffcolsprime uv$ are defined in the same way as $\colorsys u$ and $\diffcols uv$, except that here, in all states we have policies that allow interference between any two agents. With \last, we denote the final state of $\diffcols{v^m_1}{v^m_2}$, and with \lastprime, the final state of $\diffcolsprime{v^m_1}{v^m_2}$. We define the observation functions as follows:  $\obs_L(\lastprime)=1$, and for all other combinations of agent $u$ and state $s$, $\obs_u(s)=0$.
\end{itemize}

\begin{figure}[h!]
\scalebox{0.83}{
   \begin{tikzpicture}[->]
  \tikzstyle{int}=[draw, minimum size=2em]
  \tikzset{round-boxed/.style={draw=black, line width=1pt, inner sep=2mm, rectangle, rounded corners}};

  \node(agts-without-h) at (-0.5,-1) { $\agts\setminus\set h$ };
  \node(L) [below of=agts-without-h] { $L$ };
  \path (agts-without-h) edge (L) ;
  \node(s0) [round-boxed, fit = (agts-without-h) (L)] {} ;

  \node(dummy) [round-boxed] at (1.55,-0.75) { $\ $ };

  \node(colorsys-1) [round-boxed] at (3.75,-0.5) {\boxSpacing$\colorsys{u_1}$\boxSpacing};
  \node(dots1) at (5.2,-0.5) {\dots};
  \node(colorsys-n) [round-boxed] at (6.6,-0.5) {\boxSpacing$\colorsys{u_n}$\boxSpacing}; 

  \node(diffcols-1) [round-boxed] at (8.75,-0.5) {\boxSpacing$\diffcols{v^1_1}{v^1_2}$\boxSpacing};
  \node(dots2) at (10.3,-0.5) {\dots};
  \node(diffcols-n) [round-boxed] at (12,-0.5) {\boxSpacing$\diffcols{v^m_1}{v^m_2}$\boxSpacing};

  \path (s0) edge node[above] {$h$} (dummy);
  \path (dummy) edge node[above] {$\actions\setminus\set{h}$} (colorsys-1);
  \path (colorsys-1) edge (dots1);
  \path (dots1) edge (colorsys-n);
  \path (colorsys-n) edge (diffcols-1);
  \path (diffcols-1) edge (dots2);
  \path (dots2) edge (diffcols-n);

  \node(colorsys-1-prime) [round-boxed] at (3.75,-2.5) {\boxSpacing$\colorsysprime{u_1}$\boxSpacing};
  \node(dots3) at (5.2,-2.5) {\dots};
  \node(colorsys-n-prime) [round-boxed] at (6.6,-2.5) {\boxSpacing$\colorsysprime{u_n}$\boxSpacing};

  \node(diffcols-1-prime) [round-boxed] at (8.75, -2.5) {\boxSpacing$\diffcolsprime{v^1_1}{v^1_2}$\boxSpacing};
  \node(dots4) at (10.3,-2.5) {\dots};
  \node(diffcols-n-prime) [round-boxed] at (12,-2.5) {\boxSpacing$\diffcolsprime{v^m_1}{v^m_2}$\boxSpacing};

  \path (s0) edge node[below] { $\actions\setminus\set{h}$ } (colorsys-1-prime);
  \path (colorsys-1-prime) edge (dots3);
  \path (dots3) edge (colorsys-n-prime);
  \path (colorsys-n-prime) edge (diffcols-1-prime);
  \path (diffcols-1-prime) edge (dots4);
  \path (dots4) edge (diffcols-n-prime);
 \end{tikzpicture}
}
 \caption{Complete system $M^G$}\label{fig:complete system MG}
 \end{figure}
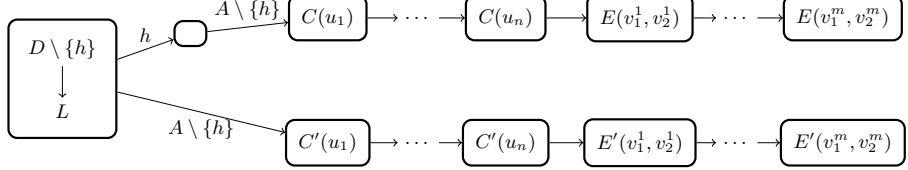

The main property of $M^G$ is that it is possible to find a path $h\alpha$ from $s_0$ to $\last$ that does \emph{not} transmit $h$ to $L$ if and only if $G$ is $3$-colorable:

\begin{definition}
A path $h\alpha$ is \emph{hiding}, if $\dom(h)\notin\dsrc{h\alpha}{L}{s_0}$, and $s_0\cdot h\alpha=\last$.
\end{definition}

Intuitively, the subsystem $\colorsys u$ forces the agent $u$ to ``choose'' a color $i\in\set{0,1,2}$, by performing the action $u_{=i}$. For each edge $(u,v)$ or $(v,u)$ in which $u$ is involved, the agent $u$ later repeats the same transition in the subsystem $\diffcols uv$ (or $\diffcols vu$). These systems ensure that no two agents that are connected with an edge can choose the same color---if they do, then a dead-end is reached. To ensure that agents are consistent in their choice of colors (i.e., choose the same color in later $\diffcols uv$-systems as in the $\colorsys u$ system, and consequently chooses the same color for each $\diffcols uv$-system), we use the following construction: When agent $u$ chooses color $i$ in $\colorsys{u}$, the agent $u_{\neq i}$ ``receives'' interference from $h$. If the agent $u$ later claims to have a color different from $i$, then the only available path is one that allows an interference between $u_{\neq i}$ and $L$, which transmits the information about $h$ to $L$.

\begin{lemma}\label{lemma:hiding path iff colorable}
 There is a hiding path if and only if $M^G$ is $3$-colarable.
\end{lemma}

\begin{proof}
 First assume that $G$ is $3$-colorable, hence let $c\colon \set{u_1,\dots,u_n}\rightarrow\set{0,1,2}$ be a coloring function such that for all edges $(u,v)\in E$, we have that $c(u)\neq c(v)$. We construct the path $a\alpha$ as the unique path from $s_0\cdot a$ to $\last$ that starts with $L$, does not use loops in any state, and where each agent $u$ chooses the action $u_{=c(u)}$ whenever the current state has more than one non-looping actions. Since $c$ is a $3$-coloring, this path does not hit a dead-end in any of the $\diffcols us$-systems, and in particular, reaches the state $\last$. Due to the construction of the path, whenever a transaction $u_{\neq i}$ is performed, the action $u_{=i}$ has never been performed on the path, and thus $u_{\neq i}$ has not received $h$. Hence none of the agents interfering with $L$ has received the action $h$, and thus $\dom(h)\notin\dsrc{a\alpha}L{s_0}$, i.e., $a\alpha$ is hiding.

 For the other direction, assume that there is a hiding path $a\alpha$. Without loss of generality, we can assume that $a\alpha$ does not use any actions that loop in the current state. Since $a\alpha$ is hiding, we know that $s_0\cdot a\alpha=\last$, in particular, every subsystem $\colorsys u$ and $\diffcols uv$ is passed when following $a\alpha$ from $s_0$. We can thus define a coloring $c\colon \set{u_1,\dots,u_n}\rightarrow\set{0,1,2}$ by $c(u)=i$, where $i$ is the unique value such that at the start of $\colorsys u$, the action $u_{=i}$ is performed by $u$. We claim that this is a $3$-coloring of $G$. 
 
 For this, first observe that on $a\alpha$, no action $u_{=j}$ is performed for $j\neq c(u)$: Due to the above, no looping action is performed. Now observe that after the performance of $u_{=c(u)}$ in $\colorsys u$, the agent $u_{\neq c(u)}$ has received the $h$-event. Now after a later performance of the action $u_{=j}$, every path that proceeds to $\last$ uses a transition $u_{\neq c(u)}$ in a state where $u_{\neq c(u)}\dintrel{} L$, which is a contradiction to the assumption that $h\alpha$ is hiding. 

 We now show that for each edge $(u,v)$ of $G$, we have that $c(u)\neq c(v)$. Since $a\alpha$ is hiding, $a\alpha$ passes through the subsystem $\diffcols uv$. Due to the above, in this subsystems the actions $u_{=c(u)}$ and $v_{=c(v)}$ are performed at the relevant states. If $c(u)$ and $c(v)$ were equal, this would reach a dead-end state, which is a contradiction, as $a\alpha$ is hiding, and hence $s_0\cdot a\alpha=\last$.
\end{proof}

Since $M^G$ can clearly be constructed from $G$ in logarithmic space, the following lemma now proves Theorem~\ref{theorem:np hardness}:

\begin{lemma}
\begin{itemize}
 \item If $G$ is $3$-colorable, then $M^G$ is not \dipsecure.
  \item If $G$ is not $3$-colorable, then $M^G$ is \dipsecure.
 \end{itemize}
 \end{lemma}

\begin{proof}
 First assume that $G$ is $3$-colorable. By Lemma~\ref{lemma:hiding path iff colorable}, there is a hiding path $h\alpha$. In particular, $s_0\cdot h\alpha=\last$. Since the action $h$ loops in the state $s_0\cdot h$, we can without loss of generality assume that $\alpha$ does not start with $h$, and hence $s_0\cdot\alpha=\lastprime$. Since $h\alpha$ is hiding, we know that $\dom(h)\notin\dsrc{h\alpha}{L}{s_0}$. Since in $s_0$, there is no outgoing edge from $h$, we also know that $\dom(h)^{s_0}_{\downarrow}\cap\dsrc{\alpha}{L}{s_0}=\emptyset$. Since $\obs_L(\last)\neq\obs_L(\lastprime)$, it follows that the $M^G$ is not 
 \dipsecure. 
 
 Now assume that $G$ is not $3$-colorable, and indirectly assume that $M^G$ is not \dipsecure.
  Since $L$ is the only agent whose observation function is not constant, this implies that there is a state $s$, an action $a$, and a sequence $\alpha$ such that $\dom(a)\notin\dsrc{a\alpha}Ls$ and $\obs_L(s\cdot a\alpha)\neq\obs_L(s\cdot\alpha)$. Since $\lastprime$ is the only state with an observation different from $0$, we know that $\lastprime\in\set{s\cdot a\alpha,s\cdot\alpha}$. In particular, $s$ is an ancestor of $\lastprime$ in $M^G$. Since $\dom(a)\notin\dsrc{a\alpha}Ls$, we know that in particular, $\dom(a)\not\dintrel{s}L$. Since the only ancestor state of $\lastprime$ in which the local policy is not the complete relation is $s_0$, we know that $s=s_0$. Since in $s_0$, all agents except for $h$ may interfere with $L$, we also know that $a=h$. Since $s_0\cdot h\alpha\neq\lastprime$ for any $\alpha$, we know that $s_0\cdot\alpha=\lastprime$. From the design of 
 $M^G$, it follows that $s_0\cdot h\alpha=\last$. Since $h\notin\dsrc{h\alpha}{L}{s_0}$, it follows that $h\alpha$ is hiding, and thus Lemma~\ref{lemma:hiding path iff colorable}, implies that $G$ is $3$-colorable as required.
\end{proof}

\begin{figure}
\scalebox{0.9}{\rotatebox{270}{%
 \begin{tikzpicture}[->,>=stealth',shorten >=1pt,auto,node distance=1cm, semithick]
     \tikzset{round-boxed/.style={draw=black, line width=1pt, dash pattern=on 1pt off 0pt, inner sep=0mm, rectangle, rounded corners}};
   \node (incoming) at (-1,-1.75) { $\cdot$ };
   \node (uneq1) at (2,3)   { $u_{\neq 1}$ };
   \node (uneq2) at (3,3)   { $u_{\neq 2}$ };
   \node (L)     at (2.5,2) { $L$ };
   \path (uneq1) edge (L);
   \path (uneq2) edge (L);
   \node (u=0) [round-boxed, fit=(uneq1) (uneq2) (L)] {} ;
   \node (u=0-checked) [round-boxed] at (6,2.5) { \phantom{\LARGE{X}} };
   \path (u=0) edge node { $u_{\neq1,2}$} (u=0-checked) ;
   \node (uneq0) at (2,-1.25)   { $u_{\neq 0}$ };
   \node (uneq2) at (3,-1.25)   { $u_{\neq 2}$ };
   \node (L)     at (2.5,-2.25) { $L$ };
   \path (uneq0) edge (L);
   \path (uneq2) edge (L);
   \node (u=1) [round-boxed, fit=(uneq0) (uneq2) (L)] {} ;
   \node (u=1-checked) [round-boxed] at (6,-1.75) { \phantom{\LARGE{X}} };
   \path (u=1) edge node { $u_{\neq0,2}$ } (u=1-checked) ;
   \node (uneq0) at (2,-5.5)   { $u_{\neq 0}$ };
   \node (uneq1) at (3,-5.5)   { $u_{\neq 1}$ };
   \node (L)     at (2.5,-6.5) { $L$ };
   \path (uneq0) edge (L);
   \path (uneq1) edge (L);
   \node (u=2) [round-boxed, fit=(uneq0) (uneq1) (L)] {} ;
   \node (u=2-checked) [round-boxed] at (6,-6) { \phantom{\LARGE{X}} };
   \path (u=2) edge node { $u_{\neq0,1}$ } (u=2-checked) ;
   \path (incoming) edge node[above] { $u_{=0}$ } (u=0) ;
   \path (incoming) edge node[above] { $u_{=1}$ } (u=1) ;
   \path (incoming) edge node[above] { $u_{=2}$ } (u=2) ;
   \node (u=0-v=0) [round-boxed ] at (9.5, 4.125) {\phantom{\LARGE{X}} };
   \path (u=0-checked) edge node[above] { $v_{=0}$ } (u=0-v=0);
   \node (vneq0) at (9,3)    { $v_{\neq 0}$ };
   \node (vneq2) at (10,3)    { $v_{\neq 2}$ };
   \node (L)     at (9.5,2.25) { $L$ };
   \path (vneq0) edge (L);
   \path (vneq2) edge (L);
   \node (u=0-v=1) [round-boxed, fit=(vneq0) (vneq2) (L)] {};
   \path (u=0-checked) edge node[above] { $v_{=1}$ } (u=0-v=1);
   \node (vneq0) at (9,1.5)    { $v_{\neq 0}$ };
   \node (vneq1) at (10,1.5)    { $v_{\neq 1}$ };
   \node (L)     at (9.5,0.75) { $L$ };
   \path (vneq0) edge (L);
   \path (vneq1) edge (L);
   \node (u=0-v=2) [round-boxed, fit=(vneq0) (vneq1) (L)] {};
   \path (u=0-checked) edge node[above] { $v_{=2}$ } (u=0-v=2);
   \node (vneq1) at (9,0)    { $v_{\neq 1}$ };
   \node (vneq2) at (10,0)    { $v_{\neq 2}$ };
   \node (L)     at (9.5,-0.75) { $L$ };
   \path (vneq1) edge (L);
   \path (vneq2) edge (L);
   \node (u=1-v=0) [round-boxed, fit=(vneq1) (vneq2) (L)] {};
   \path (u=1-checked) edge node[above] { $v_{=0}$ } (u=1-v=0);
   \node (u=1-v=1) [round-boxed ] at (9.5, -1.75) {\phantom{\LARGE{X}} };
   \path (u=1-checked) edge node[above] { $v_{=1}$ } (u=1-v=1);
   \node (vneq0) at (9,-3)    { $v_{\neq 0}$ };
   \node (vneq1) at (10,-3)    { $v_{\neq 1}$ };
   \node (L)     at (9.5,-3.75) { $L$ };
   \path (vneq0) edge (L);
   \path (vneq1) edge (L);
   \node (u=1-v=2) [round-boxed, fit=(vneq0) (vneq1) (L)] {};
   \path (u=1-checked) edge node[above] { $v_{=2}$ } (u=1-v=2);
   \node (vneq1) at (9,-4.5)    { $v_{\neq 1}$ };
   \node (vneq2) at (10,-4.5)    { $v_{\neq 2}$ };
   \node (L)     at (9.5,-5.25) { $L$ };
   \path (vneq1) edge (L);
   \path (vneq2) edge (L);
   \node (u=2-v=0) [round-boxed, fit=(vneq1) (vneq2) (L)] {};
   \path (u=2-checked) edge node[above] { $v_{=0}$ } (u=2-v=0);
   \node (vneq0) at (9,-6)    { $v_{\neq 0}$ };
   \node (vneq2) at (10,-6)    { $v_{\neq 2}$ };
   \node (L)     at (9.5,-6.75) { $L$ };
   \path (vneq0) edge (L);
   \path (vneq2) edge (L);
   \node (u=2-v=1) [round-boxed, fit=(vneq0) (vneq2) (L)] {};
   \path (u=2-checked) edge node[above] { $v_{=1}$ } (u=2-v=1);
   \node (u=2-v=2) [round-boxed ] at (9.5, -8) {\phantom{\LARGE{X}} };
   \path (u=2-checked) edge node[above] { $v_{=2}$ } (u=2-v=2);
   \node (outgoing) at (18,-1.75) { $\cdot$ };
   \path (u=0-v=1) edge [bend left] node[above] {  $v_{\neq0,2}$  } (outgoing);
   \path (u=0-v=2) edge [bend left] node[above] {  $v_{\neq0,1}$  } (outgoing);
   \path (u=1-v=0) edge node[below] {  $v_{\neq1,2}$  } (outgoing);
   \path (u=1-v=2) edge node[above] {  $v_{\neq0,1}$  } (outgoing);
   \path (u=2-v=0) edge [bend right] node[above] {  $v_{\neq1,2}$  } (outgoing);
   \path (u=2-v=1) edge [bend right] node[above] {  $v_{\neq0,2}$  } (outgoing);
  \end{tikzpicture}}}
  \caption{The subsystem $\diffcols uv$}\label{fig:diffcols}
\end{figure}

We now prove the FPT result, from which the case for a logarithmic number of agents immediately follows:

\begin{proof}
 It clearly suffices to provide an FPT algorithm. Such an algorithm can be obtained by the standard dynamic programming approach, by first creating a table with an entry for   every choice $s$, $t$ and $\Dom'$, that indicates whether $s\direl{\Dom'}t$ has already been established. 
The size of the table is $2^{\card \Dom}\cdot\card{\States}^2$. 
Now initialize the table with $\card \States\cdot\card \Actions$ operations (using the \gLR property), and use the \gSC condition to add entries to the table until no changes are performed anymore. 
Then the condition \gOC can be verified by checking, for each agent $u$, and each set $\Dom'$ for which $u\in \Dom'$, whether for all $s\direl{\Dom'} t$, we have $\obs_u(s)=\obs_u(t)$. 
For each choice of $u$ and $\Dom'$, this requires $\card \States^2$ accesses to the table. Since the access to the table can be implemented in time $2^{\card \Dom}\cdot\mathtext{poly}{\card M}$, this completes the proof. 
\end{proof}

\subsection{Proof of Theorem~\ref{theorem:redundant edges}}

\begin{proof}
  Clearly, if $M$ is not \dipsecure with respect to $\dpol$, then $M$ is also not \dipsecure with respect to $\dpolprime$. 
Using induction, we can assume that $\dpolprime$ arose from $\dpol$ by removing a single intransitively useless edge $e$. 
Assume that $M$ is not \dipsecure with respect to $\dpolprime$. 
Hence there are $a\in \Actions$, $\alpha in \Actions^*$, $s \in \States$, $u \in \Dom$ such that
$\dom(a) \notin \dsrc{a \alpha}{u}{s}$ (with respect to $\dpolprime$) and 
$\obs_u(s \cdot a \alpha) \neq \obs_u(s\cdot \alpha)$. 
Since $M$ is \dipsecure, we know that $\dom(a) \in \dsrc{a \alpha}{u}{s}$ (with respect to $\dpol$).
In particular, we know that $s\cdot a\alpha \not\dintrel{u} s\cdot \alpha$. 
It follows thtat $e$ is not intransitively useless, a contradiction. 
\end{proof}

\subsection{Proof of Theorem~\ref{theorem:polynomial unwinding characterization of intransitive uniformity and security}}

The proof of this theorem highlights an interesting difference between intransitive noninterference with a global policy (IP-security) and with local policies: It can easily be shown (see~\cite{emsw11}) that if a system is not IP-secure, then there exist a ``witness'' for the insecurity consisting of a state $s$, an agent $u$, an action $a$, and a sequence $\alpha$ such that 
\begin{enumerate}
 \item $\dom(a)\notin\sources(a\alpha,u)$ and $\obs_u(s\cdot a\alpha)\neq\obs_u(s\cdot\alpha)$ (i.e., these values demonstrate insecurity of the system), and
 \item $\alpha$ contains no $b$ with $\dom(a)\dintrel{}\dom(b)$.
\end{enumerate}

Intuitively, this means that to verify insecurity, it suffices to consider sequences in which the ``secret'' action $a$ is not transmitted even one step. This feature is crucial for the polynomial-time algorithm in~\cite{emsw11} to verify IP-security. In a setting with local policies, the situation is different, the above-mentioned property does not hold. This is in fact the key reason why no ``small'' unwinding for \dipsecurity exists, and why the verification problem is \NP-hard. However, in systems with a uniform policy, we again can prove an analogous property, even though the proof is more complicated than for the setting with a global policy:

\begin{lemma}
\label{lemma:minimal path different observations}
 Let $M$ be a system with a policy that is intransitively uniform. Then $M$ is \dipsecure if and only if there are $a$, $u$, $s$, and $\alpha$ with $\dom(a)\notin\dsrc{a\alpha}us$, $\obs_u(s\cdot\alpha)\neq\obs_u(s\cdot a\alpha)$, and no $b$ with $\dom(a)\dintrel{s}\dom(b)$ appears in $\alpha$.
\end{lemma}

\begin{proof}
 Clearly if such $a$, $u$, $s$, and $\alpha$ exist, then the system is not \dipsecure. For the converse, let $\alpha$ be of minimal length such that there exist $u$, $s$, and $a$ with $\dom(a)\notin\dsrc{a\alpha}us$ and $\obs_u(s\cdot a\alpha)\neq\obs_u(a\cdot\alpha)$. 
 Indirectly, assume that $\alpha=\beta b\beta'$ for some $b$ with $\dom(a)\dintrel{s}\dom(b)$. We consider three cases.

 \begin{itemize}
  \item \emph{Assume $\obs_u(s\cdot a\beta b\beta')\neq \obs_u(s\cdot a\beta\beta')$.} Note that $\dom(b)\notin\dsrc{b\beta'}u{s\cdot a\beta}$. Hence choosing $s'=s\cdot a\beta$, $a'=b$, and $\alpha'=\beta'$ is a contradiction to the minimality of $\alpha$.
  \item \emph{Assume $\obs_u(s\cdot\beta b\beta')\neq \obs_u(s\cdot\beta\beta')$.} To show that this again is a contradiction to the minimality of $\alpha$ (starting in the state $s\cdot\beta$), it suffices to show that $\dom(b)\notin\dsrc{b\beta'}u{s\cdot\beta}$. 
  Hence, indirectly assume that $\dom(b)\in\dsrc{b\beta'}u{s\cdot\beta}$, and let $\gamma$ be a minimal prefix of $b\beta'$ such that there is some agent $v$ with 
  \begin{itemize}
     \item $\dom(b)\in\dsrc\gamma v{s\cdot\beta}$,
     \item $\dom(a)\notin\dsrc{a\beta\gamma}vs$.
  \end{itemize}
  Since choosing $v=u$ and $\gamma=\beta'$ satisfies these conditions, such a minimal $\gamma$ exists. Again, consider the point where $v$ ``learns'' that $a$ was performed, i.e., let $\gamma=\pi c\pi'$ with 
  \begin{itemize}
    \item $\dom(b)\in\dsrc\pi{\dom(c)}{s\cdot\beta}$, and 
    \item $\dom(c)\dintrel{s\cdot\beta\pi}v$.
  \end{itemize}
  Since $\dom(a)\notin\dsrc{a\cdot\beta\gamma}vs$, and $\pi$ is a prefix of $\gamma$, the prerequisites to the lemma imply that $v^\uparrow_{s\cdot a\beta\pi}=v^\uparrow_{s\cdot\beta\pi}$, in particular, $\dom(c)\dintrel{s\cdot a\beta\pi}v$. Since $\dom(a)\notin\dsrc{a\beta\gamma}vs$, this implies 
$$\dom(a)\notin\dsrc{a\beta\pi}{\dom(c)}s,$$ hence we have a contradiction to the minimality of $\gamma$.
  \item \emph{Assume $\obs_u(s\cdot a\beta b\beta')=\obs_u(s\cdot a\beta\beta')$ and $\obs_u(s\cdot\beta b\beta')=\obs_u(s\cdot\beta\beta')$.} Since $\obs_u(s\cdot a\beta b\beta')\neq \obs_u(s\cdot\beta b\beta')$, this implies $\obs_u(s\cdot a\beta\beta')\neq\obs_u(s\cdot\beta\beta')$. To obtain a contradiction to the minimality of $\alpha$, it suffices to show that $\dom(a)\notin\dsrc{a\beta\beta'}us$. Hence, indirectly assume that $\dom(a)\in\dsrc{a\beta\beta'}us$, and let $\gamma$ be a minimal prefix of $\beta'$ such that there is an agent $v$ with
  \begin{itemize}
    \item $\dom(a)\notin\dsrc{a\beta b\gamma}vs$, and
    \item $\dom(a)\in\dsrc{a\beta\gamma}vs$.
  \end{itemize}
  Since choosing $v=u$ and $\gamma=\beta'$ satisfies these conditions, such a minimal $\gamma$ exists. Now consider the step where $v$ ``learns'' $a$, which clearly happens inside $\gamma$ (as $\dom(a)\notin\dsrc{a\beta b\gamma}vs$). Hence $\gamma=\pi c\pi'$ with
  \begin{itemize}
    \item $\dom(a)\in\dsrc{a\beta\pi}{\dom(c)}s$, and
    \item $\dom(c)\dintrel{s\cdot a\beta\pi}v$.
  \end{itemize}

 Since $\dom(a)\notin\dsrc{a\beta b\gamma}vs$, we have $\dom(b)\notin\dsrc{b\gamma}v{s\cdot a\beta}$. Since $\pi$ is a prefix of $\gamma$, this implies $\dom(b)\notin\dsrc{b\pi}v{s\cdot a\beta}$. The conditions of the lemma this imply that $v^\uparrow_{s\cdot a\beta b\pi}=v^\uparrow_{s\cdot a\beta\pi}$. In particular, this implies $\dom(c)\dintrel{s\cdot a\beta b\pi}v$. Since $\dom(a)\notin\dsrc{a\beta b\gamma}vs$, this implies $\dom(a)\notin\dsrc{a\beta b\pi}{\dom(c)}s$, which is a contradiction to the minimality of $\gamma$.
 \end{itemize}
\end{proof}

We now show a similar fact which allows us to easily verify whether a policy is intransitively uniform: To verify uniformity, it again suffices to consider action sequences in which the ``secret'' action is not even transmitted a single step. This is shown in the following Lemma:

\begin{lemma}\label{lemma:minimal path different policies}
 If a policy for a system is not intransitively uniform, there is an agent $u$, an action $a$, a sequence $\alpha$, and a state $s$ such that 
 \begin{enumerate}
  \item $\dom(a)\notin\dsrc{a\alpha}us$,
  \item $\infagents{u}{s\cdot a\alpha}\neq \infagents{u}{s\cdot\alpha}$,
 \end{enumerate}
 and contains no $b$ with $\dom(a)\dintrel{s}\dom(b)$.
\end{lemma}

\begin{proof}
 Choose $u$, $a$, $s$, and $\alpha$ such that $\card{\alpha}$ is minimal, and indirectly assume that $\alpha=\beta b\beta'$ for some sequences $\beta$ and $\beta'$, where $\dom(a)\dintrel{s}\dom(b)$. Note that this implies

 $$\dom(b)\notin\dsrc{b\beta'}u{s\cdot a\beta},$$ which we will use throughout the proof.
 We consider three cases:

 \begin{itemize}
  \item \emph{Assume that $\infagents{u}{s\cdot a\beta b\beta'}\neq \infagents{u}{s\cdot a\beta\beta'}$.}
   We choose $s'=s\cdot a\beta$, $a'=b$, and $\alpha'=\beta'$. This is a contradiction to the minimality of $\alpha$, since $\card{\alpha'}<\card{\beta'}$.

  \item \emph{Assume that $\infagents{u}{s\cdot\beta b\beta'}\neq \infagents{u}{s\cdot\beta\beta'}$.}
   We choose $s'=s\cdot\beta$, $a'=b$, and $\alpha=\beta'$ and obtain a contradiction in the same way as in the above case. For this, it suffices to prove that $\dom(b)\notin\dsrc{b\beta'}u{s\cdot\beta}$. Hence assume indirectly that $\dom(b)\in\dsrc{b\beta'}u{s\cdot\beta}$. Let $\gamma$ be a minimal prefix of $b\beta'$ such that there is an agent $v$ with 
   \begin{itemize}
     \item $\dom(b)\in\dsrc{\gamma}{v}{s\cdot\beta}$,
     \item $\dom(a)\notin\dsrc{a\beta\gamma}vs$.
   \end{itemize}
   Since $\gamma=b\beta'$ and $v=u$ satisfies these conditions, such a minimal choice of $\gamma$ and $v$ exists. Now consider the position where $v$ ``learns'' $b$, i.e., let $\gamma=\pi c\pi'$ such that the action $c$ transmits the $b$-action to $v$, i.e., we have that
   \begin{itemize}
     \item $\dom(b)\in\dsrc{\pi}{\dom(c)}{s\cdot\beta}$,
     \item $\dom(c)\dintrel{s\cdot\beta\pi}v$.
   \end{itemize}
   Note that $\pi$ is a proper prefix of $\gamma$. Since $\dom(a)\notin\dsrc{a\beta\gamma}vs$, it follows that $\dom(a)\notin\dsrc{a\beta\pi}vs$. Hence we know by the minimality of $\alpha$ that $v^\uparrow_{s\cdot\beta\pi}=v^\uparrow_{s\cdot a\beta\pi}$,
   In particular, $\dom(c)\dintrel{s\cdot a\beta\pi}v$.
   We now have the following:
   \begin{itemize}
    \item Due to the above, we know that $\dom(b)\in\dsrc{\pi}{\dom(c)}{s\cdot\beta}$,
    \item since $\dom(a)\notin\dsrc{a\beta\gamma}vs$, we know that $\dom(a)\notin\dsrc{a\beta\pi}{\dom(c)}s$.
   \end{itemize}
   Since $\pi$ is a proper prefix of $\gamma$, this is a contradiction to the minimality of $\gamma$.

   \item \emph{Assume that $\infagents{u}{s\cdot a\beta b\beta'}=\infagents{u}{s\cdot a\beta\beta'}$ and $\infagents{u}{s\cdot\beta b\beta'}=\infagents{u}{s\cdot\beta\beta'}$.} Since $\infagents{u}{s\cdot a\beta b\beta'}\neq \infagents{u}{s\cdot\beta b\beta'}$, it then follows that $\infagents{u}{s\cdot a\beta\beta'}\neq \infagents{u}{s\cdot\beta\beta'}$. It suffices to show that $\dom(a)\notin\dsrc{a\beta\beta'}{u}{s}$, we then have a contradiction to the minimality of $\alpha$. Hence indirectly assume that $\dom(a)\in\dsrc{a\beta\beta'}us$. 
Let $\gamma$ be a minimal prefix of $\beta'$ such that there is some $v$ such that
   \begin{itemize}
     \item $\dom(a)\notin\dsrc{a\beta b\gamma}vs$,
     \item $\dom(a)\in\dsrc{a\beta\gamma}vs$.
   \end{itemize}
   Since $\gamma=\beta'$ and $v=u$ satisfy these conditions, such a minimal choice exists. Similarly as before, look at the action where $a$ is forwared to $v$, i.e., let $\gamma=\pi c\pi'$ such that 
   \begin{itemize}
     \item $\dom(a)\in\dsrc{a\beta\pi}{\dom(c)}s$,
     \item $\dom(c)\dintrel{s\cdot a\beta\pi}v$.
   \end{itemize}
   Since $\dom(a)\notin\dsrc{a\beta b\gamma}vs$ and $\dom(a)\dintrel{s}\dom(b)$, it follows that $\dom(b)\notin\dsrc{b\gamma}v{s\cdot a\beta}$. Since $\pi$ is a prefix of $\gamma$, this implies $\dom(b)\notin\dsrc{b\pi}{v}{s\cdot a\beta}$. The minimality of $\alpha$ implies that 
   $v^\uparrow_{s\cdot a\beta b\pi}=v^\uparrow_{s\cdot a\beta\pi}$, in particular, $\dom(c) \dintrel{s\cdot a\beta b\pi} v$. Since $\dom(a)\notin\dsrc{a\beta b\gamma}vs$, we obtain
   \begin{itemize}
     \item $\dom(a)\notin\dsrc{a\beta b\pi}{\dom(c)}s$,
     \item from the above, we know that $\dom(a)\in\dsrc{a\beta\pi}{\dom(c)}{s}$.
   \end{itemize}
   This contradicts the minimality of $\gamma$, since $\pi$ is a proper prefix of $\gamma$.
 \end{itemize}
\end{proof}

Using these lemmas, we can now prove Theorem~\ref{theorem:polynomial unwinding characterization of intransitive uniformity and security}:

\begin{proof}
 \begin{enumerate}
  \item  First assume that there is a uniform intransitive unwinding satisfying \dPC, \dSC, and \dLR, and indirectly assume that the policy is not intransitively uniform. Due to Lemma~\ref{lemma:minimal path different policies}, there exist $a,u,s$, and $\alpha$ such that $\dom(a)\notin\dsrc{a\alpha}u{s}$, $\infagents{u}{s\cdot a\alpha}\neq \infagents{u}{s\cdot\alpha}$, and $\alpha$ does not contain any $b$ with $\dom(a)\dintrel{s}\dom(b)$. Let $v=\dom(a)$. Let $\sim^{s,v}_u$ be an equivalence relation satisfying \dPC, \dSC, and \dLR. It suffices to show that $s\cdot a\alpha\sim^{s,v}_us\cdot\alpha$ to obtain a contradiction to \dPC.

 Clearly, $\dom(a)\not\dintrel{s}u$, hence \dLR\ implies $s\sim^{s,\dom(a)}_us\cdot a$, i.e., $s^{s,v}_us\cdot a$. Note that for all $a'$ appearing in $\alpha$, we have that $\dom(a)\not\dintrel{s}\dom(a')$. Hence applying \dSC\ for each $a'$, we obtain $s\cdot a\alpha\sim^{s,v}_us\cdot\alpha$ as required.

 For the converse, assume that for all $\dom(a)\notin\dsrc{a\alpha}us$, we have that $\infagents{u}{s\cdot a\alpha}=\infagents{u}{s\cdot\alpha}$, and let $s_0$ be a state, and let $v$ and $u$ be agents. We define

 \medskip

 \begin{tabular}{llp{8cm}}
  $s\sim^{s_0,v}_ut$ & iff & for all sequences $\alpha$ that contain no $b$ with $v\dintrel{s_0}\dom(b)$, we have $\infagents{u}{s\cdot\alpha}=\infagents{u}{t\cdot\alpha}$.
 \end{tabular}

 \medskip
 Clearly, $\sim^{s_0,v}_u$ is an equivalence relation and satisfies \dPC\ (choose $\alpha=\epsilon$). For showing \dSC, let $s\sim^{s_0,v}_ut$, and let $v\not\dintrel{s_0}\dom(a)$. To show the required condition $s\cdot a\sim^{s_0,v}_ut\cdot a$, let $\alpha$ be a sequence containing no $b$ with $v\dintrel{s_0}b$. Since $v\not\dintrel{s_0}\dom(a)$, the sequence $a\alpha$ satisfies the same condition, and hence from $s\sim^{s_0,v}_ut$, it follows that $\infagents{u}{s\cdot a\alpha}=v^\uparrow_{s\cdot a\alpha}$ as required.

 Finally, consider \dLR. Let $\dom(a)\not\dintrel{s}u$. To show that $s\sim^{s,\dom(a)}_us\cdot a$, let $\alpha$ be such that no $b$ with $\dom(a)\dintrel{s}\dom(b)$ appears in $\alpha$, we need to show that $\infagents{u}{s\cdot\alpha}=\infagents{u}{s\cdot a\alpha}$. This follows from the prerequites, since clearly, $\dom(a)\notin\dsrc{a\alpha}us$.

 \item
 \begin{enumerate}
  \item Assume that the system is \dipsecure. Let $s_0$ be a state, and let $v$ and $u$ be agents. We define:

 \medskip
 \begin{tabular}{llp{8cm}}
  $s\sim^{s_0,v}_ut$ & iff & for all sequences $\alpha$ that contain no $b$ with $v\dintrel{s_0}\dom(b)$, we have $\obs_u(s\cdot\alpha)=\obs_u(t\cdot\alpha)$. 
 \end{tabular}

 \medskip
Clearly, $\sim^{s_0,v}_u$ is an equivalence relation and satisfies \dOC\ (choose $\alpha=\epsilon$). For showing \dSC, let $s\sim^{s_0,v}_ut$, and let $a\in A$ with $v\not\dintrel{s_0}\dom(a)$. We need to show that for all $\alpha$ containing no $b$ with $v\dintrel{s_0}\dom(b)$, we have $\obs_u(s\cdot a\alpha)=\obs_u(t\cdot a\alpha)$. This trivially follows from $s\sim^{s_0,v}_ut$, since $\alpha'=a\alpha$ also does not contain a $b$ with $v\dintrel{s_0}\dom(b)$.

 Finally, consider \dLR. Let $\dom(a)\not\dintrel{s}u$. We need to show that $s\sim^{s,\dom(a)}_us\cdot a$. Hence let $\alpha$ be a sequence containing no $b$ with $\dom(a)\dintrel{s}\dom(b)$. We need to show that $\obs_u(s\cdot\alpha)=\obs_u(s\cdot a\alpha)$. Since the system is \dipsecure, it suffices to show that $\dom(a)\notin\dsrc{a\alpha}us$. This follows trivially since $\dom(a)\not\dintrel{s}u$, and $\alpha$ does not contain any $b$ with $\dom(a)\dintrel{s}\dom(b)$.
 
  \item Assume that the system is not \dipsecure. Due to Lemma~\ref{lemma:minimal path different observations}, there is a state $s$, an agent $u$, an action $a$ and a sequence $\alpha$ with $\dom(a)\notin\dsrc{a\alpha}us$, $\obs_u(s\cdot a\alpha)\neq\obs_u(s\cdot\alpha)$, and $\alpha$ does not contain any $b$ with $\dom(a)\dintrel{s}\dom(b)$. Let $v=\dom(a)$, and let $\sim^{s,v}_u$ be an equivalence relation on $S$ that satisfies \dOC, \dSC, and \dLR. It suffices to show that $s\alpha\sim^{s,v}_us\cdot a\alpha$. Clearly we have that $v\not\dintrel{s}u$. Therefore, (recall that $v=\dom(a)$), \dLR implies $s\sim^{s,v}_us\cdot a$. Note that for all $b\in\alpha$, we have that $\dom(a)\not\dintrel{s}\dom(b)$. Hence applying \dSC repeatedly, we obtain $s\cdot a\alpha\sim^{s,v}_us\cdot\alpha$, which completes the proof.
 \end{enumerate}
\end{enumerate}
\end{proof}

\end{document}